\documentclass[11pt]{amsart}

\numberwithin{equation}{section}

\usepackage{float}
\usepackage{url,bm,amssymb,amscd,amsmath}
\usepackage{graphics,epsfig,color}
\usepackage[letterpaper,asymmetric,left=1in,right=1in,top=1.2in,bottom=1.
25in,bindingoffset=0.0in]{geometry}
\usepackage{mathtools}
\usepackage{amssymb}
\usepackage{amsfonts}
\usepackage{amsmath}
\usepackage{amsthm}
\usepackage{listings}
\usepackage[utf8]{inputenc}
\usepackage[english]{babel}
\usepackage{hyperref}
\newtheorem*{theorem*}{Theorem}
\newtheorem{theorem}{Theorem}[section]
\newtheorem{lemma}[theorem]{Lemma}
\newtheorem{proposition}[theorem]{Proposition}

\newtheorem{corollary}[theorem]{Corollary}

\newtheorem{remark}[theorem]{Remark}
\newtheorem*{THMA}{Theorem A}
\newtheorem*{THMB}{Theorem B}
\newtheorem*{THMC}{Theorem C}
\newtheorem*{THMC'}{Theorem C'}

\newtheorem*{SILVERMANTHME}{Silverman's Theorem E \cite{silverman}}
\newtheorem*{IRRETUTTE}{Irreducibility of Tutte Polynomials (Merino-Mier-Noy \cite{merino})}
\newtheorem*{CSS}{Call-Silverman Specialization \cite[Theorem 4.1]{callsilverman}}
\newtheorem*{CPL}{Dujardin-Favre Classification of Passivity Locus \cite[Theorem 4]{DF}}
\newtheorem*{convention}{Convention}
\newtheorem*{algpointsaredense}{Algebraic Points Are Dense}

\usepackage{tabularx,ragged2e,booktabs,caption}
\usepackage{enumerate}
\usepackage[numbers]{natbib}
\usepackage{caption}
\usepackage{subcaption}
\usepackage{mathrsfs}

\renewcommand{\hat}{\widehat}

\newcommand{\dist}{{\rm dist}_{\mathbb{P}^1}}
\newcommand{\VDEG}{V_{\rm deg}}
\newcommand{\VGOOD}{V_{\rm good}}

\newcommand{\Qbar}{\overline{\mathbb{Q}}}
\newcommand{\U}{\mathcal{U}}
\newcommand{\V}{\mathcal{V}}
\renewcommand{\P}{\mathcal{P}}

\newcommand{\TUTTE}{\mathcal{T}}

\begin{document}
	
\title{Chromatic Zeros On Hierarchical Lattices\\ and
Equidistribution on Parameter Space}

\begin{author}[I. Chio]{Ivan Chio}
	\email{ivanchio@rochester.edu}
	\address{ %
		801 Hylan Hall \\ 
		University of Rochester \\ Rochester, NY 14627 \\
		United States
		 }
\end{author}

\begin{author}[R. K. W. Roeder]{Roland K. W. Roeder}
	\email{roederr@iupui.edu}
	\address{ %
		IUPUI Department of Mathematical Sciences\\
		LD Building, Room 224Q\\
		402 North Blackford Street\\
		Indianapolis, Indiana 46202-3267\\
		United States }
\end{author}

\date{\today}

\begin{abstract} Associated to any finite simple graph $\Gamma$ is the
{\em chromatic polynomial} $\P_\Gamma(q)$ whose complex zeros are called the {\em
chromatic zeros} of $\Gamma$.  A hierarchical lattice is a sequence of finite
simple graphs $\{\Gamma_n\}_{n=0}^\infty$ built recursively using a
substitution rule expressed in terms of a generating graph.  For each $n$, let
$\mu_n$ denote the probability measure that assigns a Dirac measure to each
chromatic zero of $\Gamma_n$. Under a mild hypothesis on the generating graph,
we prove that the sequence $\mu_n$ converges to some measure $\mu$ as $n$ tends
to infinity.  We call $\mu$ the {\em limiting measure of chromatic zeros} associated
to $\{\Gamma_n\}_{n=0}^\infty$.
In the case of the Diamond Hierarchical Lattice
we prove that the support of $\mu$ has Hausdorff dimension two.

The main techniques used come from holomorphic dynamics and more specifically
the theories of activity/bifurcation currents and arithmetic dynamics.  We
prove a new equidistribution theorem that can be used to relate the chromatic
zeros of a hierarchical lattice to the activity current of a particular marked
point.  We expect that this equidistribution theorem will have several other
applications.
\end{abstract}

\maketitle

\section{Introduction}

Motivated by a concrete problem from combinatorics and mathematical physics, we
will prove a general theorem about the equidistribution of certain parameter
values for algebraic families of rational maps.  We will begin with the
motivating problem about chromatic zeros (Section \ref{SEC:INTR_CHROMATIC}) and
then present the general equidistribution theorem (Section \ref{SEC:INTRO_EQUIDISTRIBUTION}).

\subsection{Chromatic zeros on hierarchical lattices}\label{SEC:INTR_CHROMATIC}
Let $\Gamma$ be a finite simple graph.  The {\em chromatic polynomial}
$\P_\Gamma (q)$ counts the number of ways to color the vertices of $\Gamma$ with
$q$ colors so that no two adjacent vertices have the same color.  It is
straightforward to check that the chromatic polynomial is monic, has integer
coefficients, and has degree equal to the number of vertices of $\Gamma$.  The
chromatic polynomial was introduced in 1912 by G.D.  Birkhoff in an attempt to
solve the Four Color Problem \cite{birkhoff, birkhoff1}.  Although the Four
Color Theorem was proved later by different means, chromatic polynomials and
their zeros have become a central part of combinatorics.\footnote{For example, a search on Mathscinet yields
333 papers having the words ``chromatic polynomial'' in the title.}  For a comprehensive
discussion of chromatic polynomials we refer the reader to the book
\cite{DKT_BOOK}.

A further motivation for study of the chromatic polynomials comes from
statistical physics because of the connection between the chromatic polynomial
and the partition function of the antiferromagnetic Potts Model; see, for
example, \cite{WUSURVEY,shrock,sokalsurvey} and \cite[p.323-325]{BAXTER}.

We will call a sequence of finite simple graphs $\Gamma_n = (V_n, E_n)$, where
the number of vertices  $|V_n| \rightarrow \infty$, a ``lattice''.  
The standard example is the $\mathbb{Z}^d$ lattice where, for
each $n \geq 0$, one defines $\Gamma_n$ to be the graph whose vertices consist
of the integer points in $[-n,n]^d$ and whose edges connect vertices
at distance one in~$\mathbb{R}^d$.
For a given
lattice, $\{\Gamma_n\}_{n=1}^\infty$, we are interested in whether the sequence
of measures
\begin{equation}\label{limitingmeasure}
\mu_n:=\frac{1}{|V_n|} \sum_{\substack{q \in \mathbb{C} \\ \P_{\Gamma_n}(q)=0}} \delta_q
\end{equation}
has a limit $\mu$, and in describing its limit if it has one.  Here, $\delta_q$ is the Dirac measure
which, by definition, assigns measure $1$ to a set containing $q$ and measure $0$ otherwise.  (In (\ref{limitingmeasure}) zeros of $\P_{\Gamma_n}(q)$ are counted
with multiplicity.)  If $\mu$ exists,
we call it the {\em limiting measure of chromatic zeros} for the lattice
$\{\Gamma_n\}_{n=1}^\infty$.

This problem has received considerable interest from the physics community
especially through the work of Shrock with and collaborators Biggs, Chang, and Tsai
(see \cite{shrock3,shrock2,BS,CS,ST} for a sample) and Sokal with collaborators Jackson, Procacci, Salas and others
(see \cite{SS1,SS6,JPS} for a sample).  Indeed, one of the main
motivations of these papers is understanding the possible ground states
(temperature $T=0$) for the thermodynamic limit of the Potts Model, as well as
the phase transitions between them.  Most of these papers consider sequences of
$m \times n$ grid graphs with $m \leq 30$ fixed and $n \rightarrow \infty$.
This allows the authors to use transfer matrices and the Beraha-Kahane-Weiss
Theorem \cite{BKW} to rigorously deduce (for fixed $m$) properties of the
limiting measure of chromatic zeros.  The zeros typically
accumulate to some real-algebraic curves in $\mathbb{C}$ whose complexity
increases as $m$ does; see \cite[Figures 1 and 2]{shrock3} and \cite[Figures
21 and 22]{SS1} as examples.  Indeed, this behavior was first observed in the 1972  work of Biggs-Damerell-Sands \cite{BDS} 
 and then, more extensively, in the 1997 work of Shrock-Tsai \cite{ST97}.  Beyond these cases with $m$ fixed, numerical techniques are used in
\cite{SS6}  to make conjectures about the limiting behavior of the zeros as $m
\rightarrow \infty$, i.e.\ for the $\mathbb{Z}^2$ lattice.  

To the best of our knowledge, it is an open and very difficult question whether
there is a limiting measure of chromatic zeros for the $\mathbb{Z}^2$ lattice.
If such a measure does exist, rigorously determining its properties also seems quite challenging.
For this reason, we will consider the limiting measure of chromatic zeros for
\textit{hierarchical lattices}.  They are constructed as follows: start with a
finite simple graph $\Gamma \equiv \Gamma_1$ as the \textit{generating graph},
with two vertices labeled $a$ and $b$, such that $\Gamma$ is symmetric over $a$
and $b$.  For each $n > 1$, $\Gamma_n$ retains the two marked vertices $a$ and
$b$ from $\Gamma$, and we inductively obtain $\Gamma_{n+1}$ by replacing each
edge of $\Gamma$ with $\Gamma_n$, using $\Gamma_n$'s marked vertices as if they
were endpoints of that edge. {A key example to keep in mind is the Diamond
Hierarchical Lattice (DHL) shown in Figure \ref{lattice}.  In fact, one can interpret
the DHL as an anisotropic version of the $\mathbb{Z}^2$ lattice; see
\cite[Appendix E.4]{BLR1} for more details.

\begin{figure}[h!]
	\centering

	\scalebox{1.1}{
		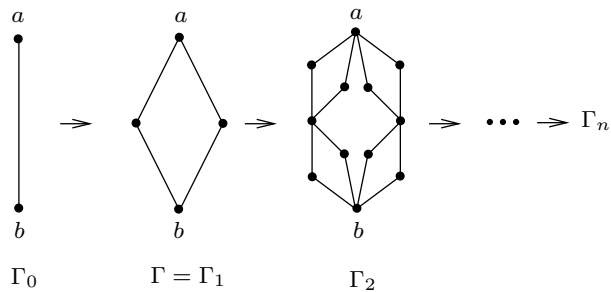
	}
	\caption{Diamond Hierarchical Lattice (DHL)}
	\label{lattice}
\end{figure}

\noindent
Several other possible generating graphs are shown in Figure \ref{generating_graphs},
including a generalization of the DHL called the $k$-fold DHL.

Statistical physics on hierarchical lattices dates back to the work of Berker and Ostlund \cite{BO}, followed by
Griffiths and Kaufman \cite{GK}, Derrida, De Seze, and Itzykson \cite{DDI}, 
Bleher and \v{Z}alys \cite{BZ1,BZ2,BZ3}, and Bleher and Lyubich \cite{BL}.

A graph $\Gamma$ is called \textit{$2$-connected} if $\Gamma$ has three or more vertices and if  there is no vertex whose removal disconnects the graph.
Our main results about the limiting measure of chromatic zeros are:

\begin{figure}[h!]
        \centering

        \scalebox{1.1}{
                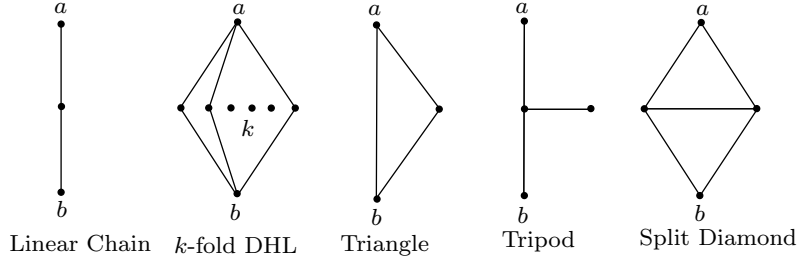
        }
        \caption{Several possible generating graphs.  The $k$-fold DHL, Triangle, and Split Diamond are $2$-connected,
while the others are not.}
        \label{generating_graphs}
\end{figure}

\begin{THMA}
Let $\{\Gamma_n\}_{n=1}^\infty$ be a hierarchical lattice whose generating
graph $\Gamma \equiv \Gamma_1$ is $2$-connected. Then its limiting measure
$\mu$ of chromatic zeros exists.
\end{THMA}

\begin{THMB}
Let $\mu$ be the limiting measure of chromatic zeros for the $k$-fold DHL and suppose $k \geq 2$.
Then, ${\rm supp}(\mu)$ has Hausdorff dimension $2$.
\end{THMB}

\begin{figure}[h!]
	\centering
	\includegraphics[scale=1.1]{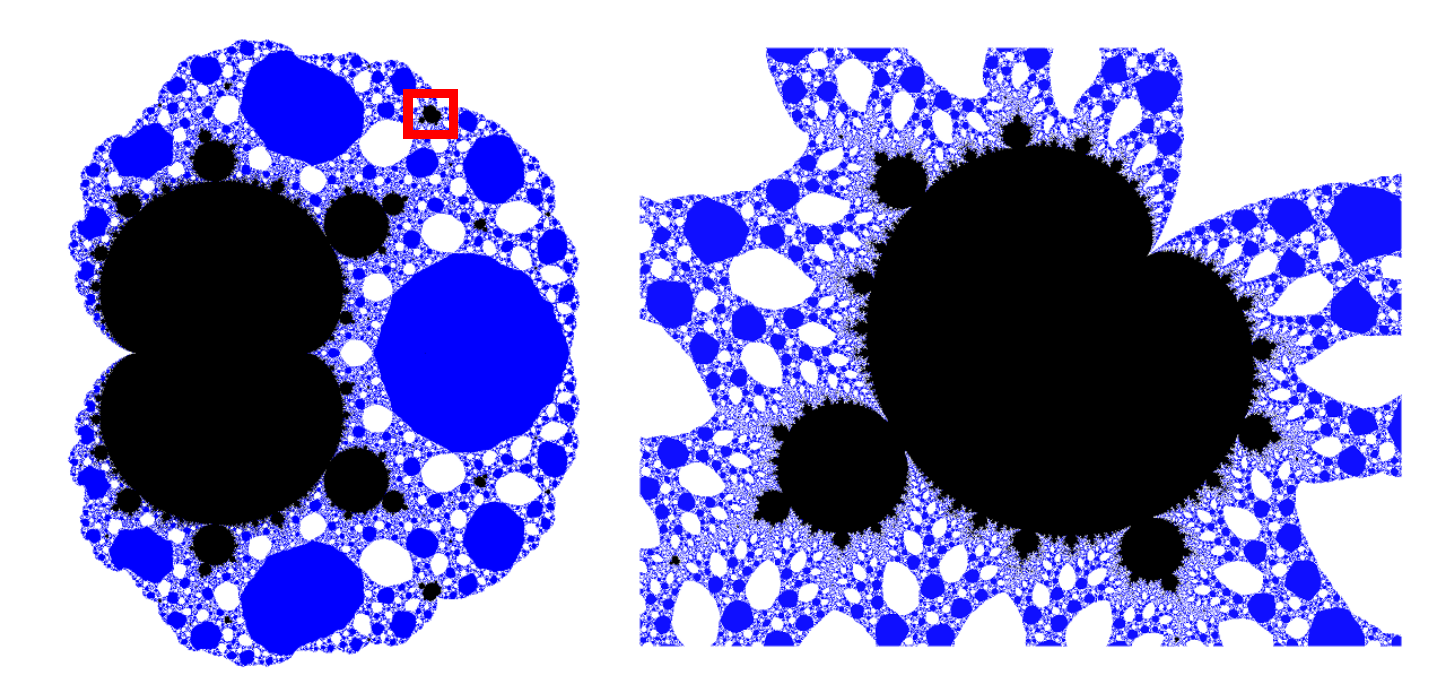}

	\caption{The support of the limiting measure of chromatic zeros for the
DHL equals the union of boundaries of the black, blue, and white sets.
Let $r_q(y)$ be the renormalization mapping for the DHL, given in (\ref{EQN:RFORDHL}).
Points in white correspond to parameter values $q$ for which $r_q^n(0)
\rightarrow 1$, points in blue correspond to parameter values $q$ for which
$r_q^n(0) \rightarrow \infty$, and points in black correspond to parameter
values for which $r_q^n(0)$ does neither. The region depicted on the left is
approximately $-2 \leq \textrm{Re} (q) \leq 4$ and $-3 \leq \textrm{Im} (q)
\leq 3$.  The region on the right is a zoomed in view of the region shown in
the red box on the left.  See Section \ref{SEC:PROOF_THMB} for an explanation of the 
appearance of ``baby Mandelbrot sets'', as on the right. 
Their appearance will imply Theorem B.
Figures 3 and 4 were made using the Fractalstream software \cite{FRACTAL}.
\label{FIG:DHLSUPP}}
\end{figure}

\begin{figure}[h!]
        \centering
        \includegraphics[scale=0.4]{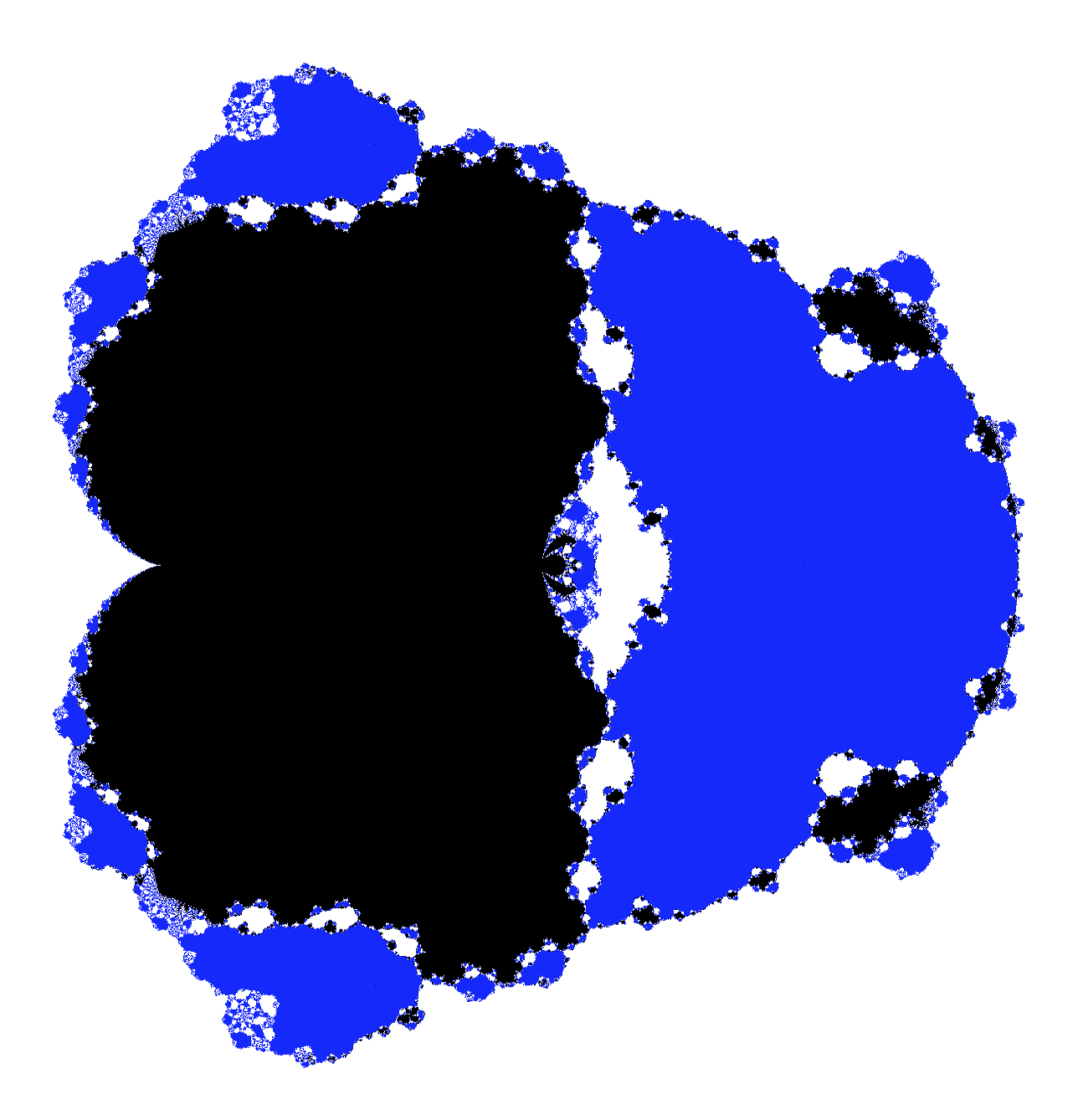}
        \caption{The support of the limiting measure of chromatic zeros for the
hierarchical lattice generated by the split diamond (see Figure \ref{generating_graphs}) equals the union of boundaries of the black, blue, and white sets.  Let $r_q(y)$ denote the 
renormalization mapping generated by the split diamond, given in (\ref{EQN:RENORM_SPLIT_DIAMOND}).  
The coloring scheme is the same as in Figure \ref{FIG:DHLSUPP}, but using this different mapping.
The region depicted is
approximately $-1 \leq \textrm{Re} (q) \leq 5$ and $-3 \leq \textrm{Im} (q)
\leq 3$.  
\label{FIG:SPLIT_DIAMOND}}
\end{figure}

\begin{remark}  The set shown in
Figure \ref{FIG:DHLSUPP} has been studied from the perspective of holomorphic dynamics  by several authors.  
We refer the reader to the works of Luo \cite{JLUO}, Aspenberg-Yampolsky \cite{MR2480740}, Wang-Qiu-Yin-Qiao-Gao \cite{MR2608338}, 
and Yang-Zeng \cite{MR3153590} for  details.
\end{remark}

The technique for proving Theorems A and B comes from the connection between the
antiferromagnetic Potts model in statistical physics and the chromatic polynomial; see, for example,
 \cite{WUSURVEY,shrock,sokalsurvey} and \cite[p.323-325]{BAXTER}.

For any graph $\Gamma$, let $Z_\Gamma(q,y)$ be the partition function (\ref{DEF:PARTI}) for the
antiferromagnetic Potts model with $q$ states and ``temperature'' $y$. We remark that $Z_\Gamma$ is defined with multivariate edge variables $y_e$, here we set $y_e = y$ for all edges, and $Z_\Gamma (q,y)$ becomes a polynomial
in both $q$ and $y$ by the Fortuin-Kasteleyn \cite{FK} representation~(\ref{DEF:PARTI2}). Then, by setting $y = 0$, one has:
\begin{align}\label{EQN:POTTS_CHROMATIC}
	\P_\Gamma(q) = Z_\Gamma(q,0).
\end{align}
See Section
\ref{PottsModel} for more details.

Given a hierarchical lattice $\{\Gamma_n\}_{n=1}^\infty$ generated by $\Gamma =
(V,E)$ let us write $Z_n(q,y) \equiv Z_{\Gamma_n}(q,y)$ for each $n \in
\mathbb{N}$.  The zero locus of $Z_n(q,y)$ is a (potentially reducible)
algebraic curve in $\mathbb{C}^2$.  However, we will consider it as a divisor
by assigning positive integer multiplicities to each irreducible component
according to the order at which $Z_\Gamma(q,y)$ vanishes on that component.  This
divisor will be denoted by
\begin{align*}
\mathcal{S}_n \, :=\, (Z_n(q, y)=0),
\end{align*}
where, in general, the zero divisor of a polynomial $p(x,y)$ will be denoted by $(p(x,y)=0)$.  Since $\Gamma_0$ is a single edge with its two endpoints,
we have
\begin{align*}
\mathcal{S}_0 \,=\, (q(y+q-1)=0) \, =\,(q=0)+(y+q-1=0).
\end{align*}
If $\Gamma$ is $2$-connected, there is a 
Migdal-Kadanoff renormalization procedure that takes 
$\Gamma$ and produces a rational map 
\begin{align*}
R \equiv R_\Gamma : \mathbb{C} \times \mathbb{P}^1 \rightarrow \mathbb{C} \times \mathbb{P}^1 \qquad  \mbox{be given by} \qquad  R(q,y) = (q,r_q(y)),
\end{align*}
with the property that
\begin{align}\label{EQN:PULLBACK}
\mathcal{S}_{n+1} = R^* \mathcal{S}_{n} \qquad \mbox{for $n\geq 0$}.
\end{align}
Here, $\mathbb{P}^1$ denotes Riemann Sphere, $R^*$ denotes the pullback of
divisors, and $r_q: \mathbb{P}^1 \rightarrow \mathbb{P}^1$ is a degree $|E|$
rational map\footnote {Actually, the degree can drop below $|E|$ for finitely
many values of $q$.} depending on $q$.  (Informally, one can think of the
pullback on divisors $R^*$ as being like the set-theoretic preimage, but
designed to keep track of multiplicities.) The reader should keep in mind the
case of the DHL for which
\begin{align}\label{EQN:RFORDHL}
r_q(y)=\left(\frac{y^2+q-1}{2y+q-2}\right)^2.
\end{align}
It will be derived in Section \ref{PottsModel}.

Because $R(q,y) = (q,r_q(y))$ is a skew product over the identity, for each $n \geq 0$ we have
\begin{align*}
\mathcal{S}_n \,\, =\,\, (R^n)^* \big((q=0)+(y+q-1=0)\big) \,\, = \,\,(q=0) + (R^n)^*(y+q-1=0).
\end{align*}
The chromatic polynomial of a connected graph $\Gamma$ has a simple zero at
$q=0$, which we can ignore when discussing the limiting measure of chromatic
zeros.  It corresponds to the divisor $(q=0)$ above.  Therefore, using
(\ref{EQN:POTTS_CHROMATIC}), all of the chromatic zeros for
$\Gamma_n$ (other than $q=0$) are given by
\begin{align}\label{EQN:CHROMATIC_RENORM1}
\tilde{\mathcal{C}}_n = (R^{n})^* (y+q-1=0) \,  \cap \, (y=0),
\end{align}
where each intersection point is assigned its Bezout multiplicity.  

Since we want to normalize and then take limits as $n$ tends to infinity, we re-write 
(\ref{EQN:CHROMATIC_RENORM1}) in terms of currents (see \cite{SIBONYDYNAMICS,DS1} for background).
We find
\begin{align}\label{EQN:CHROMATIC_RENORM2}
\tilde{\mu}_n := \frac{1}{|V_n|} [\tilde{\mathcal{C}}_n] = (\pi_{1})_* \left(\frac{1}{|V_n|} (R^{n})^* [y+q-1 =0] \ \wedge \ [y=0]\right),
\end{align}
where $\pi_1: \mathbb{C} \times \mathbb{P}^1 \rightarrow \mathbb{C}$ defined by $\pi_1(q,y) = q$ is the projection map, the square brackets denote the current of integration over a divisor, and
$\wedge$ denotes the wedge product of currents.  Since $[y = 0]$ is the current
of integration over a horizontal line, the wedge product is just
the horizontal slice of $\frac{1}{|V_n|} (R^{n})^* [y+q-1 =0]$ at height
$y=0$.  Since the wedge product results in a measure on $\mathbb{C} \times \mathbb{P}^1$, we compose with the projection $(\pi_{1})_*$ 
to obtain a measure on $\mathbb{C}$.   (In the previous two paragraphs we have used tildes on $\tilde{\mathcal{C}}_n$ and $\tilde{\mu}_n$
to denote that we have dropped the simple zero at $q=0$.)

If the generating graph $\Gamma$ is $2$-connected, then we will see in Proposition \ref{nonexceptional} that there are at most
finitely many parameters $q$ such that $y=1-q$ is an exceptional point for
$r_q$. It then  follows quickly from the one-dimensional equidistribution
theorems of Lyubich \cite{LEP, Ly2} and Freire-Lopez-Ma\~{n}\'{e} \cite{FLM}
that the following convergence holds:
\begin{align*}
\frac{1}{|E_n|} (R^{n})^* [y+q-1 =0] \rightarrow \hat{T},
\end{align*}
 where $\hat{T}$ is the {\em fiber-wise Green current} for the family of rational maps $r_q(y)$.
In Proposition \ref{vertexedge} we'll see that $\alpha:= \lim_{n\rightarrow \infty} \frac{|E_n|}{|V_n|}$ exists so that  
\begin{align*}
\hat{T}_n:=\frac{1}{|V_n|} (R^{n})^* [y+q-1 =0] \rightarrow \alpha \hat{T}.
\end{align*}
However: 

\vspace{0.1in}
\noindent
{\bf First Main Technical Issue:}  $\hat{T}_n \rightarrow \alpha \hat{T}$ does not necessarily imply $\hat{T}_n \wedge [y=0] \rightarrow 
\alpha \hat{T} \wedge [y=0]$.

\vspace{0.1in}
\noindent
This issue will be handled using the notion of activity currents which were
introduced by DeMarco in \cite{Dem1} to study bifurcations in families of
rational maps (they are sometimes called bifurcation currents).  Since then,
they have been studied by Berteloot, DeMarco, Dujardin, Favre, Gauthier, Okuyama and
many others.   We refer the reader to the surveys by Berteloot \cite{Bert1} and
Dujardin \cite{Duj1} for further details.

We can re-write (\ref{EQN:CHROMATIC_RENORM2}) as
\begin{align*}
\tilde{\mu}_n := \frac{1}{|V_n|} [(r_q^n \circ a)(q) = b(q)],
\end{align*}
where $a,b: \mathbb{C} \rightarrow \mathbb{P}^1$ are the two marked points
\begin{align*}
a(q) = 0 \qquad \mbox{and} \qquad b(q) = 1-q.
\end{align*}
(Special care must be taken at the finitely many parameters $q$ for which ${\rm
deg}_y(r_q(y)) < |E|$.  It is the Second Main Technical Issue for proving Theorem
A and it will be explained in the next subsection.)

Meanwhile, the {\em activity current} of the marked point $a$ is defined by
\begin{align*}
T_a := \lim_{n \rightarrow \infty} \frac{1}{|E_n|} (r_q^n \circ a)^* \hat{\omega},
\end{align*}
where $\hat \omega$ is the fiberwise Fubini-Study $(1,1)$ form on $\mathbb{C} \times \mathbb{P}^1$.
Therefore, proving Theorem A reduces to proving the convergence
\begin{align}\label{EQN:DESIRED_CONV}
\tilde{\mu}_n  = \frac{1}{|V_n|} [(r_q^n \circ a)(q) = b(q)] \rightarrow \alpha T_a.
\end{align}
It will be a consequence of Theorems C and C' that are presented in the next subsection.

\subsection{Equidistribution in parameter space}\label{SEC:INTRO_EQUIDISTRIBUTION}

Let $V$ be a connected projective algebraic manifold.
An {\em algebraic family of rational maps of degree $d$} is a rational mapping
\begin{align*}
f: V \times \mathbb{P}^1 \dashrightarrow \mathbb{P}^1
\end{align*}
such that, there exists an algebraic hypersurface $\VDEG \subset V$ (possibly reducible) 
with the property that for each $\lambda \in V \setminus \VDEG$ the mapping
\begin{align*}
f_\lambda: \mathbb{P}^1 \rightarrow \mathbb{P}^1 \quad \mbox{defined by} \quad f_\lambda(z) = f(\lambda,z)
\end{align*}
is a rational map of degree $d$.  
A {\em marked point} is a rational map $a: V \dashrightarrow \mathbb{P}^1$.  (We will
denote the indeterminacy locus of $a$ by $I(a)$.  It is a proper subvariety of codimension at least two.)

Our result will depend heavily on a theorem from arithmetic dynamics due to
Silverman \cite[Theorem E]{silverman} and this will require us to assume that
the manifold $V$, the family $f$, and the marked points $a$ and $b$ are defined
over the algebraic numbers $\Qbar$.  In other words, every polynomial in the
definitions of these objects has coefficients in $\Qbar$.

\begin{convention}
Throughout the paper an algebraic family of rational maps $f: V \times \mathbb{P}^1 \dashrightarrow \mathbb{P}^1$ defined over $\Qbar$ will mean that both $V$ and $f$ are defined over $\Qbar$.
\end{convention}

\begin{THMC}
	Let $f: V \times \mathbb{P}^1 \dashrightarrow \mathbb{P}^1$ be an algebraic family of rational maps of degree $d \geq 2$ defined over $\Qbar$
and let $a, b: V \dashrightarrow \mathbb{P}^1$ be two marked points defined over $\Qbar$.
Extending $\VDEG$, if necessary, we can suppose $I(a) \cup I(b) \subset \VDEG$.

Suppose that:
	\begin{enumerate}
		\item[\rm (i)]  There is no iterate $n$ satisfying $f_\lambda^n a(\lambda) \equiv b(\lambda)$.
		\item[\rm (ii)] The marked point $b(\lambda)$ is not persistently exceptional for $f_\lambda$.
	\end{enumerate}
	Then we have the following convergence of currents on $V \setminus \VDEG$
	\begin{equation}\label{EQN:THMCSTATEMENT_CONV}
	\frac{1}{d^n} \left[(f_\lambda^n \circ a)(\lambda) = b(\lambda) \right] \rightarrow T_a,
	\end{equation}
	where $T_a$ is the activity current of the marked point $a(\lambda)$.	
\end{THMC}
\noindent
The precise definition of activity current will be given in Section \ref{basicsinactivitycurrents}.

The following version of Theorem C holds on all of $V$, without removing $\VDEG$, an essential feature
for our application to Theorem A.

\begin{THMC'}
        Let $f: V \times \mathbb{P}^1 \dashrightarrow \mathbb{P}^1$ be an algebraic family of rational maps of degree $d \geq 2$ defined over $\Qbar$
and let $a, b: V \dashrightarrow \mathbb{P}^1$ be two marked points defined over $\Qbar$.
Suppose that
        \begin{enumerate}
                \item[\rm (i)]  There is no iterate $n$ satisfying $f_\lambda^n a(\lambda) \equiv b(\lambda)$.
                \item[\rm (ii)] The marked point $b(\lambda)$ is not persistently exceptional for $f_\lambda$.
        \end{enumerate}
Consider the rational map
\begin{align*}
F: V \times \mathbb{P}^1 \dashrightarrow V \times \mathbb{P}^1  \quad \mbox{defined by} \quad
F(\lambda,z) = (\lambda,f(\lambda,z)).
\end{align*}
Then the following sequence of currents on $V$
        \begin{equation}\label{EQN:THMCPRIME_CONV}
        (\pi_{1})_* \left(\frac{1}{d^n} (F^{n})^*\left[z=b(\lambda)\right] \ \wedge \ \left[z = a(\lambda)\right]\right)
        \end{equation}
converges and the limit equals 
$T_a$ when restricted to $V \setminus \VDEG$.
Here, $\pi_1: V \times \mathbb{P}^1 \rightarrow \mathbb{P}^1$ is the projection
onto the first coordinate $\pi_1(\lambda,z) = \lambda$.
\end{THMC'}

\begin{remark}
We have phrased Theorems C and C' in their natural level of generality.
However, in most applications that we have in mind (in particular to the
chromatic zeros), one can use $V = \mathbb{P}^m$ and define everything in the
usual affine coordinates $\mathbb{C}^m \subset \mathbb{P}^m$ in the following
ways:
\begin{itemize}
\item[(i)]  $f(\lambda,z) = \frac{P(\lambda,z)}{Q(\lambda,z)}$ with $P,Q \in \Qbar[\lambda,z]$ and having
no common factors of positive degree in $\Qbar[\lambda,z]$, and
\item[(ii)] $a(\lambda) = \frac{R(\lambda)}{S(\lambda)}$
with $R,S \in \Qbar[\lambda]$ and having no common factors of positive degree in $\Qbar[\lambda]$ (and similarly
for $b(\lambda)$).
\end{itemize}
The reader can keep in mind the simple case of the renormalization mapping for the DHL (\ref{EQN:RFORDHL}) in which case everything is defined over $\mathbb{Q} \subset \Qbar$.  Here $V = \mathbb{P}^1$,
\begin{itemize}
\item[(i)] $r(q,y) = \left(\frac{y^2+q-1}{2y+q-2}\right)^2$,
\item[(ii)] $a(q) \equiv 0$, and $b(q) = 1-q$.
\end{itemize}
The degree of this family is $d=4$ and $\VDEG = \{0,\infty\}$ because the degree of $r_q(y)$ drops
when $q=0$ and $q=\infty$ but at no other values of $q$.
\end{remark}

The proofs of Theorem C and C' will  closely follow the strategy that
Dujardin-Favre use in \cite[Theorem 4.2]{DF}.  However:

\vspace{0.1in}
\noindent
{\bf Second Main Technical Issue:} The proof of \cite[Theorem 4.2]{DF} requires
a technical ``Hypothesis~(H)'' that is not satisfied for the Migdal-Kadanoff renormalization mapping (\ref{EQN:RFORDHL})
for the DHL (and presumably not satisfied for many other hierarchical lattices).  Indeed, 
$q=0 \in \VDEG$ for this mapping  and there are active parameters accumulating to $q=0$.  
One sees this in Figure \ref{FIG:DHLSUPP} where $q=0$ is the ``main cusp'' on the left side of the black region.

\vspace{0.1in}
\noindent
Our assumption that the family and the marked points are defined over $\Qbar$ allows
us to avoid Hypothesis (H).  Note that, using quite different techniques, Okuyama
has proved in \cite[Theorem~1]{okuyama} a version of \cite[Theorem 4.2]{DF} without Hypothesis (H).  His proof
requires the marked point to be critical, but does not require working over $\Qbar$.  

Once Theorem C is proved, one can extend the convergence (\ref{EQN:THMCSTATEMENT_CONV})
across $\VDEG$ by an application of the compactness theorem for families of
plurisubharmonic functions \cite[Theorem 4.1.9]{hormander}, thus proving
Theorem C'.  Note that a similar statement to Theorem C' is found in  the work
of Gauthier-Vigny \cite[Corollary 3.1]{GV}.  The proof there also uses such
compactness to extend a given convergence across various ``bad'' parameters that are
analogous to our $\VDEG$.

\subsection{Brief History of Migdal-Kadanoff Renormalization}\label{SEC:HISTORY}
The renormalization mapping $r_q(y)$ given in (\ref{EQN:RFORDHL}) for the DHL
and its variants for other hierarchical lattices date back to the early 1980s.
More specifically, Migdal \cite{MIGDAL1,MIGDAL2} and Kadanoff
\cite{KADANOFF1976} described approximate renormalization equations for the
Ising Model on the $\mathbb{Z}^d$.  Berker-Ostland \cite{BO}, Bleher-\v{Z}alys
\cite{BZ3}, Kaufman-Griffiths \cite{GK}, Derrida-De Seze-Itzykson \cite{DDI},
Kinzel-Domany \cite{PhysRevB.23.3421},  Andelman-Berker \cite{PhysRevB.29.2630}, and others noticed that these
equations became exact on suitable hierarchical lattices and that the setting extends to the Potts model.  
Equation (\ref{EQN:RFORDHL}) plays a prominent role in several of the papers referenced above.
Study of Potts Models on Hierarchical Lattices continues to be an active area of physics \cite{PhysRevB.80.134201}.

\subsection{Recent works on interplay between holomorphic dynamics and statistical physics}
The present work lies in the context of several recent papers where holomorphic dynamics has played
a role in studying problems from statistical physics.  We describe a sample of them here.

To the best of our knowledge each of the previous works mentioned in Section \ref{SEC:HISTORY} focuses on zeros of
the partition function in the complex temperature plane (or sometimes the
complex magnetic field plane) but not in the $q$-plane for fixed complex
temperature.  
Studying the zeros of the partition function in the complex $q$-plane requires
quite different techniques.  To the best of our knowledge, the first time they
were studied for hierarchical lattices is by Royle-Sokal in Appendix B of \cite{ROYLESOKAL}, 
where the accumulation loci of chromatic zeros for the leaf
joined trees are studied.  Because they are interested in the accumulation loci instead
of the limiting measure of chromatic zeros, they are able to use a classical
Proposition of Lyubich \cite[Proposition 3.5]{MR751394} to deduce their results.

The novelty of our paper is that we also use holomorphic dynamics to study the
chromatic zeros, but for a different type of hierarchical lattices.  Moreover,
we are interested in the limiting measure of chromatic zeros rather than the
accumulation locus of them.  This requires us to prove a new theorem in
holomorphic dynamics (Theorem C) which can be interpreted as a quantitative
version of the Proposition  of Lyubich \cite[Proposition 3.5]{MR751394}.

One can interpret a rooted Cayley Tree as a type of hierarchical lattice, and
this allows one to apply a renormalization theory that is similar to the
Migdal-Kadanoff version used in this paper, in order to study statistical
physics on such trees.  This led to holomorphic dynamics playing an important
role in proof of the Sokal Conjecture by Peters and Regts \cite{PR1} and also
in their work on the location of Lee-Yang zeros for bounded degree graphs
\cite{PR2}.  The same renormalization theory was also recently used in
combination with techniques from dynamical systems by He, Ji, and the authors
of the present paper to characterize the limiting measure of Lee-Yang zeros for
the Cayley Tree \cite{CHJR}.

Meanwhile, holomorphic dynamics has been used by Bleher, Lyubich, and the second author of the present
paper to characterize the limiting measure of Lee-Yang zeros for the DHL \cite{BLR1} and also
to describe the limit behavior of the Lee-Yang-Fisher zeros for the DHL \cite{BLR2}.

Finally, let us note that in the recent paper \cite{MR4118578} Chang-Roeder-Shrock study the accumulation
loci of $q$-plane zeros for the Diamond Hierarchical lattice and various fixed values of temperature $y$, both in
ferromagnetic and antiferromagnetic regimes.  The key technique in that paper is again \cite[Proposition 3.5]{MR751394}.

\subsection{Relationship to a conjecture of Sokal}
For any connected graph  $\Gamma$ let $\Delta(\Gamma)$ denote the maximal degree of
a vertex of $\Gamma$.  There is a conjecture of Sokal which asserts that $|P_\Gamma(q)| > 0$ for all complex $q$ satisfying ${\rm Re}(q) > \Delta(\Gamma)$.  (See, for example, \cite[Conjecture 21]{JACKSON_CONJ}.)  

The techniques in our paper do not give insight into this conjecture because
our hypothesis that the generating graph be $2$-connected leads to the marked
vertices $a$ and $b$ having degree two or larger.  This results in 
$\Delta(\Gamma_n)$ becoming unbounded as $n$ tends to infinity.  We use $2$-connectivity of the
generating graph to guarantee that the renormalization mapping does not have
common factors (of positive degree) in the numerator and denominator.  (See Sections \ref{SEC:MK_RENORM_DHL} and  \ref{SUBSEC:MK_ARBITRARY}.)  We do
not presently see how to work around this hypothesis, however it may be quite
interesting for future study.

An additional challenge is that our techniques are about the limiting measure
of chromatic zeros and hence would not detect regions in the $q$ plane where there are a
negligible proportion chromatic zeros, in the limit as $n$ tends to infinity.

\subsection{Structure of the paper}
In Section \ref{basicsinactivitycurrents} we present background on activity
currents and describe the Dujardin-Favre classification of the passive locus,
that will play an important role in the proofs of Theorems C and C'.
Theorems C and C' are proved in
Section \ref{SEC:THMC} and Section \ref{SEC:PROOF_THMC'}.

We return to the problem of chromatic zeros in
Section \ref{PottsModel} by 
providing background on their connection with the Potts Model from statistical physics.
We also set up the renormalization mapping $r_q(y)$ associated to any hierarchical lattice having $2$-connected
generating graph.
We prove Theorem A in Section \ref{SEC:PROOF_THMA} by verifying the hypotheses of Theorem C'.

For the $k$-fold DHL with $k \geq 2$, one can check that the critical points  $y= \pm \sqrt{1-q}$
satisfy $r_q(\pm \sqrt{1-q}) \equiv 0 \equiv a(q)$.  Therefore, a result of
McMullen \cite[Corollary 1.6]{McM1} gives that ${\rm supp}(T_a)$ has Hausdorff
dimension $2$.  This is explained in Section \ref{SEC:PROOF_THMB}, where we
prove Theorem B.

We conclude the paper with Section \ref{SEC:EXAMPLES} were we 
discuss the chromatic zeros associated with the hierarchical lattices generated by each of the
graphs shown in Figure \ref{generating_graphs}.  We also provide a more detailed
explanation of Figures \ref{FIG:DHLSUPP} and \ref{FIG:SPLIT_DIAMOND}.


\vspace{0.1in}
\noindent
{\bf Acknowledgments:}
We are very grateful to Robert Shrock for introducing us to the problem of
understanding the limiting measure of chromatic zeros for a lattice and for
several helpful comments about our paper.  We are also very grateful to Laura
DeMarco and Niki Myrto Mavraki who have given us guidance on arithmetic
dynamics and also provided the details from Subsection~4.2 (Arithmetic proof of
Proposition 4.2) as well as Proposition \ref{PROP:SYMMETRY}.  We also thank
Romain Dujardin, Charles Favre, Thomas Gauthier, and Juan Rivera-Letelier for
interesting discussions and comments.  We thank the anonymous referees for
their detailed reading of our paper and for several comments and suggestions,
which have helped to considerably improve the paper.  This work was supported by NSF grant
DMS-1348589.




\section{Basics in Activity Currents}\label{basicsinactivitycurrents}

\subsection{Holomorphic Families, Marked Points, and Active/Passive Dichotomy}

Let $\Lambda$ be a connected complex manifold. A {\em holomorphic family of
rational maps of degree $d \geq 2$} is a holomorphic map $f:\Lambda \times
\mathbb{P}^1 \rightarrow \mathbb{P}^1$ such that $f_{\lambda}:=f(\lambda,
\cdot): \mathbb{P}^1 \rightarrow \mathbb{P}^1$ is a rational map of degree $d$
for every $\lambda \in \Lambda$.  

Associated with $f_\lambda$ is the skew product mapping
\begin{align}
F: \Lambda \times \mathbb{P}^1 \rightarrow \Lambda \times \mathbb{P}^1  \quad  \mbox{given by} \quad
F(\lambda,z) = (\lambda,f_\lambda(z)).
\end{align}
Note that it is conventional in the literature to denote by $f_\lambda^n(z)$ the second component of $F^n(\lambda,z)$.

A {\em marked point} is a holomorphic map $a:
\Lambda \rightarrow \mathbb{P}^1$.
The marked point $a:\Lambda \rightarrow \mathbb{P}^1$ is called {\em passive}
at $\lambda_0 \in \Lambda$ if the family $\{f^n_\lambda a(\lambda) \}$ is
normal in some neighborhood of $\lambda_0$, otherwise $a$ is said to be {\em
active} at $\lambda_0$. The set of all parameters where $a$ is active is called the {\em active
locus of $a$}.  

\begin{remark}
Historically in holomorphic dynamics one considers {\em marked critical
points}, i.e. marked points $a(\lambda)$ that are critical points of
$f_\lambda$ for every parameter $\lambda$.  In this paper it will be crucial to
consider non-critical marked points.  Fortunately, several results from the
classical literature carry over to our setting.  We will review in this section
the results that we need and carefully check that the marked points need not be
critical.

Let us note that recently there has been considerable interest in the dynamical
properties of non-critical marked points, with some of the motivations coming
from problems in arithmetic dynamics.  As a sample of such recent papers, we
refer the reader to \cite{DeM3,DEMARCO_KAWA,GAUTHIER1,FAVRE_GAUTHIER} and the references therein.
\end{remark}

\subsection{Activity Current for Holomorphic Families}
\label{SEC:LOCAL_ACTIVE_CURRENT}
The active locus naturally supports a closed positive $(1,1)$ current $T_a$ called the {\em activity current of $a(\lambda)$}, introduced
by Laura DeMarco in \cite{DeM2}.   (We refer the reader to \cite{DEMAILLY} for background on currents, plurisubmarmonic (PSH) functions, and Monge-{A}mp\`ere operators.)

The construction of $T_a$ can be done in a coordinate-free manner, however we will first express it in local coordinates, which are simpler for explicit calculations and allow us to check that the marked point need not be critical.

Suppose $\Lambda$ is an open subset of~$\mathbb{C}^m$.
We can choose a lift
$\tilde{f}: \Lambda \times \mathbb{C}^2 \rightarrow \mathbb{C}^2$ which is holomorphic and so that for each
$\lambda \in \Lambda$,
\begin{align}\label{EQN:LIFT_OF_F}
\tilde{f}_\lambda (z, w) := \tilde{f}(\lambda, z, w)= (P_\lambda(z, w), Q_\lambda(z,w)),
\end{align}
where $P_\lambda, Q_\lambda$ are both homogeneous polynomials of degree $d$.
Similarly, the marked
point $a: \Lambda \rightarrow~\mathbb{P}^1$ can be lifted to a
holomorphic map
\begin{align*}
\tilde{a}: \Lambda \rightarrow \mathbb{C}^2 \setminus \{(0,0)\}.
\end{align*}
The choices of lifts $\tilde{f}$ and $\tilde{a}$ are unique up to a non-vanishing scaling factor that depends holomorphically on $\lambda$.

The function
$G_n: \Lambda \times \mathbb{C}^2 \rightarrow [-\infty,\infty)$ given by
\begin{align*}
G_n(\lambda,(z,w)) := \frac{1}{d^n}  \log||\tilde{f}^n_\lambda(z,w)||
\end{align*}
is PSH in $(\lambda,z,w)$, where $|| \cdot ||$ is the Euclidean norm on $\mathbb{C}^2$.  Over any compact subset of $\Lambda$ there is a constant $C > 0$ such
that 
\begin{align*}
C^{-1} ||(z,w)||^d \leq ||\tilde{f}_\lambda(z,w)|| \leq C ||(z,w)||^d.
\end{align*}
Using this one can check that the $G_n$ converge uniformly on compact subsets
of $\Lambda \times \left(\mathbb{C}^2 \setminus \{(0,0)\}\right)$.  This 
implies that the limit $G:\Lambda \times \mathbb{C}^2 \rightarrow
[-\infty,\infty)$ is PSH and continuous off of $\Lambda \times \{(0,0)\}$.

Let ${\rm pr}: \mathbb{C}^2 \setminus \{(0,0)\} \rightarrow \mathbb{P}^1$ be
the canonical projection.   For any sufficiently small open $U \subset
\mathbb{P}^1$ there is a holomorphic $s: U \rightarrow \mathbb{C}^2 \setminus
\{(0,0)\}$ such that ${\rm pr} \circ s(u) = u$ for all $u \in U$.
One defines the {\em fiberwise Green current} $\hat{T}$ on $\Lambda \times \mathbb{P}^1$ locally on $\Lambda \times U$
by 
\begin{align}\label{EQN:FIBERWISE_GREEN_LOCAL_POTENTIAL}
\hat{T} := {\rm dd}^c G(\lambda,s(u)).
\end{align}
Here, ${\rm dd}^c := \frac{i}{\pi} \partial \overline{\partial}$ is the
Monge-Amp\`{e}re operator and ${\rm d}$ is the exterior derivative (not to be
confused with the degree $d$ of a rational map).  In (\ref{EQN:FIBERWISE_GREEN_LOCAL_POTENTIAL})
${\rm dd}^c$ is taken with respect to $(\lambda,u)$.

Now consider the functions
\begin{align}\label{EQN:HN_H}
H_n(\lambda):= G_n(\lambda,\tilde{a}(\lambda)) \qquad \mbox{and} \qquad  H(\lambda):= G(\lambda,\tilde{a}(\lambda)).
\end{align} 
Both are PSH functions and the locally uniform convergence $G_n$ to $G$ described in the previous paragraph implies
locally uniform convergence of $H_n$ to $H$.
We define the {\em activity current} by
\begin{align*}
T_a := {\rm dd}^c H(\lambda).
\end{align*}

When defining $\hat{T}$ and $T_a$ we have made choices of lifts
$\tilde{f}_\lambda$, $\tilde{a}$, and $s: U \rightarrow \mathbb{C}^2 \setminus
\{(0,0)\}$.  One can check that a different choice results in adding a
pluriharmonic function to $G(\lambda,s(u))$ and/or to $H(\lambda)$ and hence
does not affect the definitions of $\hat{T}$ and $T_a$.

If $\Lambda$ is a complex manifold that is not an open subset of
$\mathbb{C}^m$ then one defines $T_a$ using the above formula in local coordinates.  The definition is compatible
under change of coordinates.

\begin{theorem}\label{suppofbifcur} {\bf (DeMarco \cite{DeM2})}  The support of the activity current $T_a$  coincides with the active locus of $a$.
\end{theorem}

\noindent
Indeed, in \cite[Theorem 9.1]{DeM2} DeMarco proves the equivalence
of the following two statements:
\begin{itemize}
\item[(i)] The functions $\{\lambda \mapsto f_\lambda^n(a(\lambda)) \, : \, n \geq 0\}$ forms a normal family in a neighborhood of $\lambda_0$.
\item[(ii)] For any holomorphic lift $\tilde{a}(\lambda)$ the function $G(\lambda,\tilde{a}(\lambda))$ is pluriharmonic in a neighborhood of $\lambda_0$.
\end{itemize}
Although we do not include the proof here, let us note that it is explicit, concise and does not require the marked point $a(\lambda)$ to be critical.

\subsection{Coordinate-free description of activity current $T_a$.}
Let $\omega$ be the Fubini-Study $(1,1)$ form on $\mathbb{P}^1$ and let $\hat
\omega = \pi_2^* \omega$, where $\pi_2(\lambda,z) = z$ is the projection onto
the second coordinate.  


\begin{proposition}\label{PROP:GREENCURRENT}
We have
\begin{align}
\hat{T} = \lim_{n \rightarrow \infty} \frac{1}{d^n} (F^n)^* \widehat{\omega} \qquad \mbox{and}  \qquad T_a = \lim_{n \rightarrow} \frac{1}{d^n} (f_\lambda^n \circ a)^* \omega,
\end{align}
where $F(\lambda,z) = (\lambda,f_\lambda(z))$ is the skew product associated with $f_\lambda$.
\end{proposition}

\begin{proof}
This is a consequence of the convergence of the potentials $G_n(\lambda,(z,w))$ to $G(\lambda,(z,w))$ that was discussed in the previous subsection and the fact that
\begin{align*}
{\rm pr}^* \omega = {\rm dd}^c \log \|(z,w)\|,
\end{align*} 
where ${\rm pr}: \mathbb{C}^2 \setminus \{(0,0)\} \rightarrow \mathbb{P}^1$ is the canonical projection; see, for example, \cite[p. 30]{GH}.
\end{proof}

There is an equivalent, alternative, description of the activity
current $T_a$ given in \cite[Proposition 3.1]{DF}.  We state it here because it ties nicely with the discussion presented in the 
introduction of our paper, but we note that it's not actually needed in our proofs of Theorems C and C'.
Let $v_n, v_\infty$ be the local potentials of $d^{-n}(F^n)^* \widehat{\omega}$
and $\widehat{T}$ respectively. In the proof of
\cite[Proposition 3.1]{DF} one sees that all the $v_n$'s and
$v_\infty$ are continuous, and particularly $v_n \rightarrow v_\infty$ locally
uniformly. (This corresponds to the properties and convergence of $G_n(\lambda,(z,w))$ to
$G(\lambda,(z,w))$ described in Section \ref{SEC:LOCAL_ACTIVE_CURRENT}.) 
We therefore have the following corollary:

\begin{corollary}\label{bifurcationdefinition}
For any marked point $a:\Lambda \rightarrow \mathbb{P}^1$, we have the following convergence of intersection of currents:
 \begin{equation}\label{EQN:CONV_TO_ACTIVITY_CURRENT_DEF}
 \frac{1}{d^n}(F^n)^* \widehat{\omega} \wedge \left[z = a(\lambda)\right] \rightarrow \widehat{T} \wedge \left[z = a(\lambda)\right].
 \end{equation}
\end{corollary}

\noindent
Let $\Gamma:=\{(\lambda, a(\lambda)\} \subset \Lambda \times \mathbb{P}^1$.
Since $\pi_1:  \Gamma \rightarrow \Lambda$ is a biholomorphism, one can check that
\begin{equation}\label{bifurcationcurrent}
T_a = (\pi_1)_*\left(\widehat{T} \wedge\left[z = a(\lambda)\right]\right).
\end{equation}

\subsection{Activity Current for Algebraic Families}
Let $f: V \times \mathbb{P}^1 \dashrightarrow \mathbb{P}^1$ be an
algebraic family of rational maps of degree $d$. 
In this case, one can delete $\VDEG$ to obtain a holomorphic family
$f: (V \setminus \VDEG) \times \mathbb{P}^1 \rightarrow \mathbb{P}^1$
and the construction from the
previous subsection defines the activity current $T_a$ for 
$f: (V \setminus \VDEG) \times  \mathbb{P}^1
\rightarrow \mathbb{P}^1$. We will now show that there is a natural extension
of $T_a$ through the hypersurface $\VDEG$. 


As the construction is local, we can suppose $V$ is a precompact open subset
of~$\mathbb{C}^m$ and proceed as in Section~\ref{SEC:LOCAL_ACTIVE_CURRENT}.  We
choose a homogeneous lift $\tilde{f}: V \times \mathbb{C}^2 \rightarrow
\mathbb{C}^2$ as in (\ref{EQN:LIFT_OF_F}) and so that
\begin{align}\label{EQN:LIFT_NORMALIZATION}
\sup_{\lambda \in V, \, \|(z,w)\| =1} |\tilde{f}_\lambda(z,w)| = 1. 
\end{align}
In this case $P_\lambda,
Q_\lambda$ have degree $d' \leq d$, with equality iff $\lambda \in V \setminus
\VDEG$.

Similarly, the marked
point $a: V \dashrightarrow \mathbb{P}^1$ can be lifted to a
holomorphic map
\begin{align*}
\tilde{a}:
V \rightarrow \mathbb{C}^2.
\end{align*}
Note that unlike in the holomorphic families we can have $\tilde{a}(\lambda) = (0,0)$ for some $\lambda \in V_{\rm deg}$,
corresponding to indeterminate points of $a$.

We now use Equation (\ref{EQN:HN_H}) to define a PSH function $H_n: V
\rightarrow [-\infty,\infty)$ using the lifts chosen above.  However, a-priori,
we only know that $H_n$ converges on $V \setminus V_{\rm deg}$ to the PSH
function $H: V \setminus V_{\rm deg} \rightarrow [-\infty,\infty)$.  
The next two
propositions show that the convergence extends (in an appropriate way) to all
of $V$.

\begin{proposition}\label{PROP:PC}
Suppose $V \subset \mathbb{C}^m$ is open.
	The pointwise limit 
	\begin{align*} 
	H(\lambda):=\lim_{n\rightarrow \infty} H_n (\lambda)
	\end{align*}
exists and is PSH in $V$. When restricted to $V \setminus
\VDEG$, the current $\tilde{T}_a := {\rm dd}^c H$ is identically equal to the activity current $T_a$. 	
\end{proposition}

\begin{proof}

Fix a parameter $\lambda \in V$ so that $f_\lambda : \mathbb{P}^1
\rightarrow \mathbb{P}^1$ is a rational map with degree $d' \leq d$. 
Using (\ref{EQN:LIFT_NORMALIZATION}) and the homogeneity of $\tilde{f}_\lambda$, 
	\begin{align*}
	|| \tilde{f}_\lambda (z, w) || \leq ||(z, w)||^{d'}, \mbox{ which implies }  
	|| \tilde{f}^{n+1}_\lambda (z, w) || \leq || \tilde{f}^n_\lambda (z, w)||^{d'}, 
	\end{align*}
	so the maps $H_n$ satisfy
	\begin{align*}
	H_{n+1} (\lambda) = \frac{1}{d^{n+1}} \log ||(\tilde{f}^{n+1}_\lambda  \circ \tilde{a}) (\lambda)||
	\leq \frac{1}{d^{n+1}} \log || (\tilde{f}^n_\lambda \circ \tilde{a}) (\lambda)||^{d'}
	\leq \frac{1}{d^n} \log ||(\tilde{f}^n_\lambda  \circ \tilde{a}) (\lambda)|| = H_n (\lambda),
	\end{align*}
which implies that $\{ H_n \}_{n=1}^\infty$ is a decreasing sequence of
PSH functions, so it either converges to a PSH limit function $H$  or to $-\infty$
identically.   The latter is impossible since $H_n(\lambda)$ converges to a finite value
for any $\lambda \in V \setminus V_{\rm deg}$.
\end{proof}

Denote by $L^1_{\rm loc}(V)$ the space of locally integrable functions in $V$.
Since the convergence given by Proposition \ref{PROP:PC} is only pointwise we will check the following.
\begin{proposition}\label{PROP:L1LOCC}
Suppose $V \subset \mathbb{C}^m$ is open.
The sequence of PSH functions $H_n$ converges to $H$ in $L^1_{\rm loc}(V)$. Equivalently, the sequence of currents ${\rm dd}^c H_n$ converges to $\tilde{T}_a = {\rm dd}^c H$.
\end{proposition}

\begin{proof}
Assume on the contrary that $H_n$ does not converge to $H$ in $L^1_{\rm loc}(V)$.
Then the compactness theorem for PSH functions \cite[Theorem 4.1.9]{hormander}
implies that there is a subsequence $H_{n_k}$ and some PSH function $H' \neq H$ in $L^1_{\rm loc}(V)$
such that $H_{n_k} \rightarrow H'$ in $L^1_{\rm loc}(V)$. 
Then there is a set $\Omega
\subset V$ of positive measure in which $H'(\lambda) \neq H
(\lambda)$ for all $\lambda$. In particular, since $\VDEG$ is a hypersurface
and hence has zero measure, we can find a compact set $K \subset \Omega
\setminus \VDEG$ in which $H_{n_k} \rightarrow H'$. However, by Corollary
\ref{bifurcationdefinition}, $H_n \rightarrow H$ uniformly in $K$, which is a
contradiction.
\end{proof}

\subsection{Classification of the Passivity Locus}

\begin{CPL}
Let $f: \Lambda \times \mathbb{P}^1 \rightarrow \mathbb{P}^1$ be a holomorphic family and let $a(\lambda)$ be a marked point. Assume $U \subset \Lambda$ is a connected open subset where $a(\lambda)$ is passive. Then exactly one of the following cases holds:
	
	\begin{enumerate}
		\item[\rm (i)] $a(\lambda)$ is never preperiodic in $U$. In this case the closure of the orbit of $a(\lambda)$ can be followed by a holomorphic motion.
		\item[\rm (ii)] $a(\lambda)$ is persistently preperiodic in $U$.

		\item[\rm (iii)] There exists a persistently attracting
		(possibly superattracting) cycle attracting $a(\lambda)$ throughout $U$
		and there is a closed subvariety $U' \subsetneq U$ such that the set of
		parameters $\lambda \in U \setminus U'$ for which $a(\lambda)$ is preperiodic
		is a proper closed subvariety in $U \setminus U'$.

		\item[\rm (iv)] There exists a persistently irrationally
		neutral periodic point such that $a(\lambda)$ lies in the interior of its
		linearization domain throughout $U$ and the set of parameters $\lambda \in U$ for
		which $a(\lambda)$ is preperiodic is a proper closed subvariety in $U$.

	\end{enumerate}
	
\end{CPL}

There are two differences between the statement above and the statement given by Dujardin and Favre in 
\cite[Theorem 4]{DF}.  In that paper:

\begin{enumerate}
\item the authors suppose the marked point is critical, and
\item in Part (iii) the authors claim that the set of parameters for which $a(\lambda)$ is preperiodic is a proper 
closed subvariety of $U$, without first removing $U'$. 
\end{enumerate}
For this reason we present the proof of the classification of the passivity locus from \cite{DF}, carefully verifying that
the marked point need not be critical. (In Remark \ref{RMK:CLARIFICATION} we will also clarify the issue about Part (iii) of the statement.)

Proof of the classification of the passivity locus is a consequence of the following theorem in local holomorphic dynamics
which we cite verbatim from \cite[Theorem 1.1]{DF}:

\begin{theorem}\label{THM:LOCAL_DYNAMICS}
Let $f_\lambda$ be any holomorphic family of holomorphic maps parameterized by a connected complex manifold $\Lambda$.  Suppose that each $f_\lambda$ is defined on the unit disc with values in $\mathbb{C}$, and leaves the origin fixed, i.e.\ $f_\lambda(0) = 0$.
Let $\lambda \mapsto p(\lambda)$ be any holomorphic map such that $p(\lambda_0) = 0$ for some parameter $\lambda_0$.

Assume that for all $n \in \mathbb{N}$, the function $f_\lambda^n p(\lambda)$ is well-defined and takes its values in the unit disk.  Then one of the following three cases holds.
\begin{enumerate}
\item For every $\lambda \in \Lambda$, the point $0$ is attracting or superattracting, and $p(\lambda)$ lies
in the (immediate) basin of attraction of $0$.
\item The point $p$ is periodic for all parameters, i.e.\ $f^{\ell}_\lambda p(\lambda) = p(\lambda)$ for some $\ell$ and all $\lambda \in \Lambda$.
\item  The multiplier of $f_\lambda$ at $0$ is constant and equals ${\rm exp}(2i\pi\theta)$, with $\theta \in \mathbb{R} \setminus \mathbb{Q}$.  The map $f_\lambda$ is linearizable and $p(\lambda)$ lies in the interior of the domain of linearization of $f_\lambda$.
\end{enumerate}
\end{theorem}

Note that in the above theorem there is no assumption about $p(\lambda)$ being
a critical point or being some iterate of a critical point.

We will also need the following lemma which is adapted from of \cite[p. 524-525]{Duj1}.

\begin{lemma}\label{LEMMA:BRANCHED_COVER}
Let $f_\lambda$ be any holomorphic family of holomorphic maps parameterized by
a connected complex manifold $\Lambda$.  Then, after passing to a branched cover
of $\Lambda$, any fixed point of $f_\lambda$ can be followed holomorphically over all of $\Lambda$.

More specifically, suppose that $z_0$ is a fixed point of
$f_{\lambda_0}$ for some $\lambda_0 \in \Lambda$.  Then, there is   
a holomorphic branched cover $\tau: \tilde{\Lambda} \rightarrow \Lambda$
and a holomorphic map $z_0: \tilde{\Lambda} \rightarrow \mathbb{P}^1$ with the following properties.
If we
let 
\begin{align}\label{EQN:FTILDE}
\tilde{f}_{\tilde{\lambda}}(z) := f_{\tau(\tilde{\lambda})}(z)
\end{align}
then 
\begin{enumerate}
\item $z_0(\tilde{\lambda})$ is a fixed point of $\tilde{f}_{\tilde{\lambda}}$ for every $\tilde{\lambda} \in \tilde{\Lambda}$ and
\item  For any $\tilde{\lambda_0}$ with $\tau(\tilde{\lambda_0}) = \lambda_0$ we have 
$z_0(\tilde{\lambda_0}) = z_0.$
\end{enumerate}
\end{lemma}

\begin{proof}
Consider the analytic hypersurface
\begin{align*}
\{(\lambda,z) \in \Lambda \times \mathbb{P}^1 \, : \, f_\lambda(z) = z\}
\end{align*}
and let $\hat{\Lambda}$ be the irreducible component containing
$(\lambda_0,z_0)$.  Let $\pi_1: \hat{\Lambda} \rightarrow \Lambda$ and $\pi_2:
\hat{\Lambda} \rightarrow \mathbb{P}^1$ be projection onto the first and second
coordinates, respectively.
A-priori, $\hat{\Lambda}$ could have singularities, but we can let
$\tilde{\Lambda}$ be the desingularization of it and let $\tau: \tilde{\Lambda}
\rightarrow \Lambda$ and $z_0: \tilde{\Lambda} \rightarrow \mathbb{P}^1$
be the lifts of $\pi_1$ and $\pi_2$, respectively, to $\tilde{\Lambda}$.
Properties (1) and (2) then follow if we let
$\tilde{f}_{\tilde{\lambda}}$ be defined as in (\ref{EQN:FTILDE}).
\end{proof}

\begin{proof}[Proof of the classification of the passive locus]
What follows is very close to what is presented in~\cite{DF}.  

Let $U \subset \Lambda$ be a connected open set on which $a(\lambda)$ is passive.
Suppose that
neither Cases (i) or (ii) of the statement hold, in order to prove that Case
(iii) or (iv) holds.  Then, there is some parameter $\lambda_0 \in U$ for which
$a(\lambda_0)$ is preperiodic for $f_{\lambda_0}$, but not persistently in $U$.  
Passing to a suitable
iterate, we can suppose that $a(\lambda_0)$ is prefixed, i.e.\ that there is
some iterate $n_0 \geq 0$ such that $f^{n_0}_{\lambda_0}(a(\lambda_0)) = w$
with $w$ a fixed point of $f_{\lambda_0}$.

A-priori the fixed point $w$ may not vary holomorphically with $\lambda$ (it is
the case when $f_{\lambda_0}'(w) = 1$).  However, it follows from Lemma
\ref{LEMMA:BRANCHED_COVER} that we can can replace $\Lambda$ by a branched
covering $\tilde{\Lambda}$ so that $w$ depends holomorphically on
$\tilde{\lambda}$ and so that $a(\tilde{\lambda})$ continues to be passive for
the lifted family $\tilde{f}_{\tilde{\lambda}}$ for all $\tilde{\lambda} \in
\tau^{-1}(U)$.  To keep notation simple we will suppose that $w$ already varied holomorphically
over $\Lambda$.

Let $v \neq w$ be a repelling fixed point for $f_{\lambda_0}$.  As in the previous paragraph, after passing to a branched cover of $\Lambda$, we can also suppose that $v$ varies holomorphically over all of $\Lambda$.

Conjugating by a suitable holomorphically varying M\"obius transformation we
can therefore suppose that $w(\lambda) = 0$ and $v(\lambda) = \infty$ for all
$\lambda \in \Lambda$.  Let $p(\lambda) := f^{n_0}_\lambda(a(\lambda))$ and
note that $p(\lambda_0) = w = 0$ but that this does not hold on all of $U$.  

We claim that for every $\lambda \in U$ and every $n \geq 0$ we have that $f_\lambda^n(p(\lambda)) \neq \infty$.
Suppose for contradiction that there is some $\lambda_1 \in U$ and some $n_1$
such that $f_{\lambda_1}^{n_1}(p(\lambda_1)) = \infty$.  Since $a(\lambda)$ and hence
$p(\lambda)$ is passive on $U$ we can then find a small neighborhood $\mathcal{N}$ of $\lambda_1$ so that
\begin{align*}
f_\lambda^{n_1+n}(p(\lambda)) \in \mathbb{D}_1(\infty)
\end{align*}
for all $\lambda \in \mathcal{N}$ and all $n \geq 0$.  (Here,
$\mathbb{D}_1(\infty)$ is the disc of radius $1$ centered at infinity.) It then
follows from Theorem \ref{THM:LOCAL_DYNAMICS} that for all $\lambda \in
\mathcal{N}$ the fixed point at infinity is attracting
or that it has a constant multiplier equal to ${\rm exp}(2i\pi
\theta)$ with $\theta \in \mathbb{R} \setminus \mathbb{Q}$.
In the first case, 
we would have that $f_\lambda^{n_1+n}(p(\lambda))~\rightarrow~\infty$ on $\mathcal{N}$.
Since $p(\lambda)$ is passive on $U \supset \mathcal{N}$ this would need to happen on all of $U$
contrary to the fact that $f_{\lambda_0}^k(p(\lambda_0)) = 0$ for all $k \geq 0$.
In the second case the multiplier of infinity at $\lambda_0$ would also
equal ${\rm exp}(2i\pi \theta)$ contrary to the hypothesis that $v(\lambda_0) = \infty$ is a repelling fixed point for $f_{\lambda_0}$.

Since $\lambda \mapsto f_\lambda^n(p(\lambda))$ forms a normal family as
mappings into $\mathbb{P}^1$ and none of the mappings hit infinity for any
$\lambda \in U$ they form a normal family of mappings into $\mathbb{C}$ (over
all $\lambda \in U$).
Normality implies that 
for any precompact open $V \subset U$ we can find a disc $D_{R}(0)$ of some radius $R > 0$
such that 
\begin{align*}
f_\lambda^{n}(p(\lambda)) \in \mathbb{D}_R(0)
\end{align*}
for all $n \geq 0$.
It then
follows from Theorem \ref{THM:LOCAL_DYNAMICS} that for all $\lambda \in
V$ the fixed point at $0$ is attracting 
or that it has a constant multiplier equal to ${\rm exp}(2i\pi \theta)$ with $\theta \in \mathbb{R} \setminus \mathbb{Q}$.
Since $V \subset U$ was arbitrary this holds on all of $U$.
In particular, we are in Cases (iii) or (iv) from the statement of the Classification of the passive
locus.

It remains to prove the statements about preperiodic parameters in these two cases.
Suppose we are in Case (iii) so that $z=0$ is attracting or superattracting for all $\lambda \in U$.
For each $\lambda \in U$ let $m(\lambda)$ denote
the local multiplicity of $0$ for $f_\lambda$.  Let $m_0:={\rm min}\{m(\lambda) \, : \, \lambda \in U\}$ and let
\begin{align*}
U':= \{\lambda \in U \, : \, m(\lambda) > m_0\}.
\end{align*}
Suppose $\lambda_0 \in U \setminus U'$, and choose a neighborhood $W$ of
$\lambda_0$ such that its closure $\overline{W}$ is compactly contained in $U \setminus U'$.  Then, there exists
$\epsilon > 0$ such that:
\begin{itemize}
\item[(A)] $f_\lambda(\mathbb{D}_\epsilon(0))$ is compactly contained in  $\mathbb{D}_\epsilon(0)$, and
\item[(B)] for each $\lambda \in W$ and each $z \in \mathbb{D}_\epsilon(0) \setminus
\{0\}$ we have that $f_\lambda(z)
\neq 0$,  \\
i.e. $0$ is the only preimage of $0$ under
$f_\lambda$ within $\mathbb{D}_\epsilon(0)$.
\end{itemize}

Since $\overline{W}$ is compact and $f_\lambda^n(a(\lambda)) \rightarrow
0$ for all $\lambda \in \overline{W}$ there exists $k > 0$ such that
for all $\lambda \in W$ we have that $f_\lambda^k(a(\lambda)) \subset
\mathbb{D}_\epsilon(0)$.  Then, using (B) above, the set of
preperiodic parameters in $W$ is
\begin{align*}
\{\lambda \in W \, : \, f_\lambda^n(a(\lambda)) = 0 \mbox{ for some $0 \leq n \leq k$}\},
\end{align*}
which is a closed subvariety of $W$.  We conclude that 
that assertions from Case (iii) of the classification of the passive locus 
hold for all $\lambda \in U$.

\vspace{0.1in}
Suppose we are in Case (3) of Theorem \ref{THM:LOCAL_DYNAMICS} so that $z=0$
has multiplier ${\rm exp}(2i\pi\theta)$, with $\theta \in \mathbb{R} \setminus
\mathbb{Q}$, The map $f_\lambda$ is linearizable and $p(\lambda)$ lies in the
interior of the domain of linearization of $f_\lambda$.  In this case, as
explained in \cite{DF}, by making $U$ smaller (if necessary), there is a
uniform $\epsilon > 0$ such that for all $\lambda \in U$ the disc
$D_\epsilon(0)$ is contained in the linearization domain for $f_\lambda$.
It then follows that $p(\lambda)$ is preperiodic for $f_\lambda$ if and only if $p(\lambda) = 0$.
Therefore, in $U$ the marked point $a(\lambda)$ is pre-periodic if and only $f^{n_0}(a(\lambda)) = 0$.  Since
$n_0$ is fixed, this is an analytic condition on $\lambda$.  We conclude that the assertions from Case (iv) of the the classification of the passive locus hold for $\lambda \in U$.
\end{proof}

\begin{remark}\label{RMK:CLARIFICATION}
One should note that the statement in \cite{DF} claims that in Case (iii) the
set $\lambda$ such that $a(\lambda)$ is preperiodic is a closed subvariety of
$U$ itself, without first removing a proper closed subvariety $U'$.
Unfortunately, that is not true, even if the marked point is critical. Fortunately this claim about preperiodic parameters is not
used anywhere later in their paper.  

Consider the following holomorphic family of polynomial mappings
\begin{align*}
f_\lambda(z) = z(z-\lambda)(z-1/2)
\end{align*}
where $\lambda \in \mathbb{C}$. 
The critical points of $f_\lambda$ are
\begin{align*}
c_\pm(\lambda) = \frac{(1+2\lambda)\pm \sqrt{4\lambda^2-2\lambda+1}}{6},
\end{align*}
which vary holomorphically in a neighborhood $\mathbb{D}_r(0)$ of $\lambda = 0$, for some $r > 0$.  Notice that $c_-(0) = 0$ and $c_+(0) = \frac{1}{3}$.
Consider the marked critical point $c(\lambda):= c_+(\lambda)$.
One can check that
\begin{enumerate}
\item There exists $0 < \epsilon < r$ such that $\lambda \in \mathbb{D}_\epsilon(0)$ implies that
$f^n_\lambda(c(\lambda)) \rightarrow 0$ with $|f^n_\lambda(c(\lambda))| < 1/2$ for all $n \geq 0$, and
\item There exists an infinite sequence $\{\lambda_k\}_{k=1}^\infty$ in $\mathbb{D}_\epsilon(0) \setminus \{0\}$ with
$\lambda_k \rightarrow 0$ such that for each $k$ there is an iterate $n_k$ with $f^{n_k}_{\lambda_k}(c(\lambda_k)) = 0$.
\end{enumerate}
Therefore, the set of preperiodic parameters $\lambda \in \mathbb{D}_\epsilon(0)$ is not a closed subvariety,
but they are in $\mathbb{D}_\epsilon(0) \setminus \{0\}$.
\end{remark}


\section{Proof of Theorem C}\label{SEC:THMC}
Our proof of Theorem C will closely follow the strategy that Dujardin-Favre use to prove Theorem~4.2 in~\cite{DF}
and we will assume some of the basic results from their proof.

\subsection{Strategy for Proof of Theorem C}

Let $f: V \times \mathbb{P}^1 \dashrightarrow \mathbb{P}^1$ be an algebraic family
of rational maps of degree $d$ defined over $\Qbar$.  Let
$a,b: V \dashrightarrow \mathbb{P}^1$ be marked points and assume, without loss of generality,
that the indeterminacy $I(a) \cup I(b) \subset \VDEG$. 
Let $T_a$ be the activity current of $a(\lambda)$ and suppose that all
hypotheses of Theorem C are satisfied.

\begin{proposition}\label{PROP:4.1}
The following convergence of currents
 \begin{equation}\label{EQN:THMC_CONV}
        \frac{1}{d^n} \left[(f_\lambda^n \circ a)(\lambda) = b(\lambda) \right] \rightarrow T_a
\end{equation}
holds in $V \setminus \VDEG$
if and only if there is a dense set of parameters $\lambda \in \VGOOD \subset V \setminus \VDEG$
such that
\begin{align}\label{EQN:GOOD_DENSE_SET}
h_n(\lambda) := \frac{1}{d^n}\log \dist \left(f^n_\lambda a(\lambda), b(\lambda)\right) \rightarrow 0,
\end{align}
where $\dist$ denotes the chordal distance on $\mathbb{P}^1$.
\end{proposition}

\begin{proof}
A direct adaptation of the first four paragraphs of the proof of Theorem 4.2 in \cite{DF}
shows that (\ref{EQN:THMC_CONV}) holds if and only if $h_n \rightarrow 0$ in $L^1_{\rm loc}(V \setminus \VDEG)$.  Let us
summarize the key steps here.

This is a local statement, so we can suppose without loss of generality that
$V$ is an open subset of $\mathbb{C}^m$.  Choose a lift $\tilde{f}: V \times
\mathbb{C}^2 \rightarrow \mathbb{C}^2$ and denote the iterates of each
$\tilde{f}_\lambda: \mathbb{C}^2 \rightarrow \mathbb{C}^2$ by
\begin{align*}
\tilde{f}_\lambda^n (z, w) = \left( P_\lambda^{(n)}(z, w), Q_\lambda^{(n)}(z, w)\right).
\end{align*}
Choose lifts $\tilde{a}, \tilde{b}: V \rightarrow \mathbb{C}^2$ and denote their coordinates by $\tilde{a}(\lambda)=(a_1(\lambda), a_2(\lambda))$ and $\tilde{b}(\lambda)=(b_1(\lambda), b_2(\lambda))$.

Using the formula for chordal distance on $\mathbb{P}^1$ we have:
\begin{align}\label{EQN:HN_VIA_LIFTS}
h_n (\lambda):= \frac{1}{d^n} \log |P_\lambda^{(n)}(\tilde{a}(\lambda))b_2(\lambda)-Q_\lambda^{(n)}(\tilde{a}(\lambda))b_1(\lambda) |^2
-\frac{1}{d^n} \log ||(\tilde{f}^n_\lambda \circ \tilde{a}) (\lambda)||^2
-\frac{1}{d^n} \log ||\tilde{b}(\lambda)||^2.
\end{align}
Note that the last term converges $0$ and the second to last term converges
to $2H(\lambda)$, both locally uniformly on $V \setminus V_{\rm deg}$.  (Here, $H(\lambda)$ is the local potential for $T_a$.)

Therefore we can conclude that $h_n \rightarrow 0$ in $L^1_{\rm loc}(V
\setminus \VDEG)$ if and only if 
\begin{align*}
\frac{1}{d^n} \log |P_\lambda^{(n)}(\tilde{a}(\lambda))b_2(\lambda)-Q_\lambda^{(n)}(\tilde{a}(\lambda))b_1(\lambda) |^2 \rightarrow 2H(\lambda) \qquad \mbox{in $L^1_{\rm loc}(V \setminus \VDEG)$.}
\end{align*}
The PSH functions on the left hand side are local potentials for the currents expressed in (\ref{EQN:THMCSTATEMENT_CONV})
and $H$ is a local potential for $T_a$ on $V \setminus \VDEG$.

\vspace{0.1in}
If $h_n$ converges to $0$ in $L^1_{\rm loc}(V \setminus
\VDEG)$ then clearly it converges on a dense set.  Conversely,
suppose $h_n$ does not converge to $0$ in $L^1_{\rm loc}(V \setminus
\VDEG)$, then, as in paragraph seven of the proof of Theorem 4.2 in \cite{DF}
shows that (\ref{EQN:THMC_CONV}), one uses Hartogs' Lemma \cite[Theorem
4.1.9(b)]{hormander} to find an open set $U \subset V \setminus \VDEG$ and a
subsequence $n_k$ such that $h_{n_k}(\lambda) \rightarrow h(\lambda) <  0$ for
all $\lambda \in U$. 
\end{proof}

The proof of Theorem C will then follow immediately from the following:

\begin{proposition}\label{PROP:GOOD_DENSE_SET}
There is a dense set of parameters $\lambda \in \VGOOD \subset V \setminus \VDEG$
such that (\ref{EQN:GOOD_DENSE_SET}) holds.
\end{proposition}

This will follow from the Dujardin-Favre classification of the passive locus and the following beautiful theorem:

\begin{SILVERMANTHME}
	Let $\phi :\mathbb{P}^1 \rightarrow \mathbb{P}^1$ be a rational map of degree $d \geq 2$
	defined over a number field $K$. Let $A, B \in \mathbb{P}^1(\overline{K})$ and assume 
	that $B$ is not exceptional for $\phi$ and that $A$ is not preperiodic for $\phi$. Then
	\begin{align}\label{EQN:SILVERMAN_LIMIT}
	\lim_{n \rightarrow \infty} \frac{\delta(\phi^n A, B)}{d^n} = 0,
	\end{align} 
where $\delta(P,Q)=2-\log \dist(P,Q)$ is the {\em logarithmic distance function}\footnote{In \cite[Theorem E]{silverman} a different logarithmic distance function was used. However, as mentioned in Section 3 of the referenced paper, the result still holds if we use $2-\dist(\cdot,\cdot)$ instead.}.
\end{SILVERMANTHME}

Remark that (\ref{EQN:SILVERMAN_LIMIT}) holds if and only if $\lim_{n \rightarrow \infty} \frac{1}{d^n}\log \dist(\phi^n A,B) = 0$.

\subsection{Proof of Proposition \ref{PROP:GOOD_DENSE_SET}}

We will need the following result:

\begin{algpointsaredense}
        Let $V \subset \mathbb{P}^n$ be a projective algebraic manifold that is defined over~$\Qbar$.
        Then, the set of points $a \in V$ that can be represented by homogeneous coordinates in $\Qbar$ form a dense subset of
        $V$ (in the complex topology).  I.e.\ $V(\Qbar)$ is dense in $V$.
\end{algpointsaredense}

We could not find this statement in the literature, but it can be proved by
induction on ${\rm dim}(V)$.  The base of the induction, when
${\rm dim}(V) = 0$, plays an important role in the theory of Kleinian Groups,
see for example \cite[Lemma 3.1.5]{MR}.

\begin{proof}[Proof of Proposition \ref{PROP:GOOD_DENSE_SET}]
We will consider the active and passive loci separately.  Let $\lambda_0$
be an active parameter, and $W \subset V$ be any open neighborhood containing
$\lambda_0$.   We will find a parameter $\lambda_1 \in W$ at which
(\ref{EQN:GOOD_DENSE_SET}) holds.
We will do this by showing that 
there exists a parameter $\lambda_1
\in W$ such that iterates of $a(\lambda_1)$ under $f_{\lambda_1}$ will eventually land on a repelling cycle
disjoint from~$b(\lambda_1)$.  This will immediately imply (\ref{EQN:GOOD_DENSE_SET}) at $\lambda_1$. 

Pick three distinct points in a repelling cycle of $f_{\lambda_0}$ which is
disjoint from $b(\lambda_0)$. By reducing $W$ to a smaller neighborhood of
$\lambda_0$ if necessary, we can ensure that the repelling cycle moves
holomorphically as $\lambda$ varies over $W$, and that $b(\lambda)$ is disjoint
from the cycle for every $\lambda \in W$. Since the family $\{f^n_{\lambda}
a(\lambda)\}_{n=1}^\infty$ is not normal in $W$, it cannot avoid all three
points.

We now suppose $\lambda_0$ is in the passive locus for $a$ and let $U$
be the connected component of the passive locus containing $\lambda_0$.  Then,
the Dujardin-Favre classification gives four possible behaviors for
$f^n_\lambda(a(\lambda))$ in $U$.

In Cases (i),(iii), and (iv) the classification gives a (possibly empty) closed
subvariety $U' \subsetneq U$ such that the set of parameters for which
$a(\lambda)$ is preperiodic is contained in a proper closed subvariety $U_1
\subset (U \setminus U')$.  Moreover, the hypothesis that marked point
$b(\lambda)$ is not persistently exceptional gives that there is another proper
closed subvariety $U_2 \subset U$ such that $b(\lambda)$ is not exceptional for
$\lambda \in U \setminus U_2$.  Then, $U \setminus (U' \cup U_1 \cup U_2)$ is
an open dense subset of $U$.  Since $V(\Qbar)$ is dense in $V$, see the beginning of this subsection,
arbitrarily close to $\lambda_0$ is a
point $\lambda_1 \in U \setminus (U' \cup U_1 \cup U_2)$ with coordinates in
$\Qbar$.  Since there are only finitely many coefficients to consider, we can
find a number field $K$ so that $f_{\lambda_1} \in K(z)$ and the points
$a(\lambda_1), b(\lambda_1) \in \mathbb{P}^1(K)$.  Since $\lambda_1 \in U
\setminus (U' \cup U_1 \cup U_2)$, the point $a(\lambda_1)$ has infinite orbit
under $f_{\lambda_1}$, and the point $b(\lambda_1)$ is not exceptional for
$f_{\lambda_1}$.  Hence Silverman's Theorem E implies that
(\ref{EQN:GOOD_DENSE_SET}) holds for the parameter $\lambda_1$.

Finally suppose we are in Case (ii), so that the marked point $a(\lambda)$ is
persistently preperiodic. By assumption there is no iterate $n$ with
$f^n_\lambda(a(\lambda)) \equiv b(\lambda)$, so there is a proper closed
subvariety $U_1 \subset U$ such that for all $\lambda \in U \setminus U_1$ and
all $n \geq 0$, we have $f^n_\lambda(a(\lambda)) \neq b(\lambda)$.  It follows
that for each $\lambda \in U \setminus U_1$, the quantities
$\dist(f^n_{\lambda}(a(\lambda)), b(\lambda))$ are uniformly bounded in $n
\geq 0$, which  implies~(\ref{EQN:GOOD_DENSE_SET}) for all $\lambda \in U \setminus U_1$.
\end{proof}

\subsection{Arithmetic proof of Proposition \ref{PROP:GOOD_DENSE_SET} under additional hypotheses}
Under additional hypotheses we can prove Proposition \ref{PROP:GOOD_DENSE_SET}
(and hence Theorems C and C') without appealing to the Dujardin-Favre classification of the
passive locus.  Instead we will require some technical results from arithmetic dynamics.

The additional hypotheses we need are:
\begin{itemize}
	\item[(iii)] The parameter space is $\mathbb{P}^1$.
	\item[(iv)]  The marked point $a$ is not passive on all of $\mathbb{P}^1 \setminus \VDEG$.
\end{itemize}
For applications in chromatic zeros our parameter space is $\mathbb{P}^1$ so that
Hypothesis (iii) will automatically hold (in fact, we typically think of it as $\mathbb{C} \subset \mathbb{P}^1$).
Meanwhile, for the renormalization mappings associated with many hierarchical lattices one can check Hypothesis (iv) 
directly, but it does not hold for all such mappings (e.g. when the generating graph is a triangle, as
discussed in Section \ref{SUBSEC:TRIANGLES}).

Proposition \ref{PROP:GOOD_DENSE_SET} will follow from Silverman's Theorem E 
and the next statement (choosing $K$ to be dense in $\mathbb{C}$), whose proof was communicated to us by Laura DeMarco and Niki Myrto Mavraki.

\begin{proposition}\label{PROP:D-M}
	Suppose the hypotheses in Theorem C and additionally hypotheses (iii) and (iv) above.
	Then, for any number field $K$ there are at most finitely many parameters $\lambda \in \mathbb{P}^1(K) \setminus \VDEG$ such that the marked point $a(\lambda)$ is preperiodic under $f_\lambda$.

\end{proposition}

We will need the following two results, which depend on having a one-dimensional parameter space.
Denote the logarithmic absolute Weil height on $\Qbar$ by $h: \Qbar \rightarrow \mathbb{R}$.
For a rational map $\phi: \mathbb{P}^1 \rightarrow \mathbb{P}^1$ defined over
$K$ and a point $P \in \mathbb{P}^1(K)$, we denote the canonical height
function associated to $\phi$ by~$\hat{h}_\phi(P)$. For more background on these definitions, see \cite{silvermanbook}. 

\begin{CSS}
	Let $(f, a)$ be a one-dimensional algebraic family of rational maps of
degree $d \geq 2$ with a marked point $a$, both defined over a number field
$K$. Then, for any sequence of parameters $\{\lambda_n\}_{n=1}^\infty \subset
\mathbb{P}^1(K) \setminus \VDEG$ such that $h(\lambda_n) \rightarrow \infty$, we
have
	\begin{equation*}
	\lim_{n\rightarrow \infty} \frac{\hat{h}_{f_{\lambda_n}}(a(\lambda_n))}{h(\lambda_n)} = \hat{h}_f (a),
	\end{equation*}
	where $\hat{h}_f (a)$ is the canonical height associated to the pair $(f,a)$. 
\end{CSS}
\noindent
The canonical height $\hat{h}_f(a)$ was introduced in \cite{callsilverman}.

The pair $(f,a)$ is called
{\em isotrivial} if there exists a branched covering $W \rightarrow \mathbb{P}^1
\setminus \VDEG$ and a family of holomorphically varying M\"obius transformations
$M_\lambda$ such that $M_\lambda\circ f_\lambda \circ M_\lambda^{-1}$ is
independent of $\lambda \in W$ and also $M_\lambda\circ a$ is a constant
function of $\lambda \in W$.

\begin{theorem}{\bf (DeMarco \cite[Theorem 1.4]{DeM3})}\label{DeM3}
	Suppose $f: \mathbb{P}^1 \times \mathbb{P}^1 \dashrightarrow \mathbb{P}^1$ is a
	non-isotrivial one-dimensional algebraic family of rational maps.
	Let $\hat{h}_f:\mathbb{P}^1(\bar{k})\rightarrow \mathbb{R}$
	be a canonical height of $f$, defined over the function field
	$k=\mathbb{C}(\mathbb{P}^1)$. For each $a \in \mathbb{P}^1(\bar{k})$, the following
	are equivalent:
	
	\begin{enumerate}
		\item[\rm (1)] The marked point $a$ is passive in all of $\mathbb{P}^1 \setminus \VDEG$;
		\item[\rm (2)] $\hat{h}_f (a) = 0$;
		\item[\rm (3)] $(f, a)$ is preperiodic.
	\end{enumerate}
	Moreover, the set 
	\begin{align*}
	\{a\in \mathbb{P}^1(k) : a \mbox{ is passive in all of } \ \mathbb{P}^1 \setminus \VDEG \}
	\end{align*}
	is finite.
\end{theorem}





\begin{proof}[Proof of Proposition \ref{PROP:D-M}]
	
	
	Assume on the contrary that there is a sequence of distinct parameters
	$\{\lambda_n\}_{n=1}^\infty \subset \mathbb{P}^1(K)  \setminus \VDEG$ 
	such that $a(\lambda_n)$ is preperiodic for $f_{\lambda_n}$.
	It follows from Northcott property \cite[Theorem 3.7]{silvermanbook}
	that the
	parameters $\lambda_n$ satisfies $h(\lambda_n) \rightarrow \infty$.
	Meanwhile, since $a(\lambda_n)$ is preperiodic for $f_{\lambda_n}$, we have
	$\hat{h}_{f_{\lambda_n}}(a(\lambda_n)) = 0$.  Then Call-Silverman
	Specialization implies  $\hat{h}_f (a) = 0$, so by Theorem \ref{DeM3} the
	marked point $a$ must be passive in all of $\mathbb{P}^1 \setminus \VDEG$,
	which contradicts hypothesis~(iv).
	
\end{proof}


\section{Proof of Theorem C'}\label{SEC:PROOF_THMC'}

The following statement about convergence of sequences of PSH functions is probably standard, but we will include a proof
because we cannot find an appropriate reference.

\begin{proposition}\label{LEM:PATCHING}
Let $\{\phi_n\}_{n=1}^\infty$ be a sequence of PSH functions in an open connected set $U \subseteq \mathbb{C}^m$ which is uniformly bounded above in compact sets. Suppose there is a PSH function $\phi$ in $U$ such that $\phi_n \rightarrow \phi$ in $L^1_{loc}(U \setminus X)$, where $X \subset U$ is an analytic hypersurface. Then $\phi_n \rightarrow \phi$ in $L^1_{loc}(U)$.
\end{proposition}

\begin{proof}
Assume by contradiction that $\phi_n$ does not converge to $\phi$ in $L^1_{loc}(U)$. Then there is an $\epsilon >0$, a compact set $K$ with positive Lebesgue measure, and a subsequence $\phi_{n_k}$ such that 
\begin{align*}
|| \phi_{n_k} - \phi ||_{L^1(K)} > \epsilon \ \mbox{ for all } k.
\end{align*} 
Note that since $\phi_n \rightarrow \phi$ in $L^1_{\rm loc}(U \setminus X)$, the compact set $K$ must intersect $X$.
By the hypotheses, the sequence $\phi_{n_k}$ satisfies the conditions for the compactness theorem for PSH functions \cite[Theorem 4.1.9]{hormander}, so we can find a further subsequence (which we still denote by $\phi_{n_k}$), and a PSH function $\tilde{\phi}$ in $U$ such that 
\begin{align*}
	\phi_{n_k} \rightarrow \tilde{\phi} \quad \mbox{in} \quad L^1_{loc}(U).
\end{align*}
In particular $\phi_{n_k} \rightarrow \tilde{\phi}$ in $L^1(K)$, which implies
$\tilde{\phi} \neq \phi$ in $L^1(K)$, so there exist $\delta > 0$ and a compact
subset $K' \subset K$ with positive Lebesgue measure such that
$|\tilde{\phi}(z) - \phi (z)| > \delta$ for all $z \in K'$. Let $X_\epsilon$ be
the $\epsilon$-neighborhood of $X$ in $U$, and let $X'_\epsilon:= X_\epsilon
\cap K'$. Choose $\epsilon > 0$ which satisfies ${\rm Leb}(K') = 2 {\rm Leb}
(X'_\epsilon)$, where ${\rm Leb}$ denotes
Lebesgue measure. It follows that
\begin{align}\label{EQU:DIFF}
\int_{K' \setminus X'_\epsilon} |\tilde{\phi} - \phi| \ d{\rm Leb} > \delta \cdot {\rm Leb}(K' \setminus X'_\epsilon) = \frac{\delta}{2}{\rm Leb}(K') > 0.
\end{align}
Meanwhile, since $K' \setminus X'_\epsilon$ is a compact subset of $U$ disjoint from
$X$, we must have $\tilde{\phi} = \phi$ in $L^1 (K' \setminus X'_\epsilon)$,
which contradicts (\ref{EQU:DIFF}).
\end{proof}

\begin{proof}[Proof of Theorem C']
This is a local statement, so we can suppose without loss of generality that
$V$ is an open subset of $\mathbb{C}^m$.
In the proof of 
Theorem C we saw that
\begin{align}\label{EQN:THMCPRIME}
\frac{1}{d^n} \log |P_\lambda^{(n)}(\tilde{a}(\lambda))b_2(\lambda)-Q_\lambda^{(n)}(\tilde{a}(\lambda))b_1(\lambda) |^2 \rightarrow 2H(\lambda) \qquad \mbox{in $L^1_{\rm loc}(V \setminus \VDEG)$.}
\end{align}
Here, the notation is as in the proof of Proposition \ref{PROP:4.1}.
Note that the PSH functions on the left hand side of (\ref{EQN:THMCPRIME}) are defined on all of $V$ and are potentials for
the currents defined in Equation~(\ref{EQN:THMCPRIME_CONV}) from Theorem C'.

On any precompact open subset of $V$ we can choose our lift $\tilde{a}(\lambda)$ sufficiently close
to the origin $(0,0) \in \mathbb{C}^2$ so that it is in the basin of attraction of $(0,0)$ under
$\tilde{f}_\lambda(z,w) = (P_\lambda(z,w),Q_\lambda(z,w))$. 
It then follows that the sequence of potentials on the left hand side of (\ref{EQN:THMCPRIME}) is locally bounded above.  Hence, Proposition \ref{LEM:PATCHING} implies that the convergence extends to all of $V$.

\end{proof}


\section{The Potts Model, Chromatic Zeros, and Migdal-Kadanoff Renormalization}\label{PottsModel}

We first give a brief account of the antiferromagnetic Potts model on a graph
$\Gamma$ and its connection with the chromatic zeros of $\P_\Gamma$. Suitable
references include  \cite{WUSURVEY,shrock,sokalsurvey}, \cite[p.323-325]{BAXTER}, and references therein.  We then describe the  Migdal-Kadanoff
Renormalization procedure that produces a rational function $r_q(y)$ relating
the zeros for the Potts Model on one level of a hierarchical lattice
to the zeros for the next level.  The remainder of the section is devoted to
proving properties of the renormalization mappings $r_q(y)$.

\subsection{Basic Setup}

Fix a graph $\Gamma = (V,E)$ and fix an integer $q \geq 2$. A \textit{spin configuration} of the graph $\Gamma$ is a map
\begin{equation*}
\sigma : V \rightarrow \{1,2,...,q\}.
\end{equation*} 
Fix the coupling constant  $J < 0$. The energy $H_\Gamma(\sigma)$ associated with a configuration $\sigma$ on $\Gamma$ is defined as
\begin{equation*}
H_\Gamma(\sigma) = -J \sum_{\{v_i,v_j\} \in E} \delta(\sigma(v_i), \sigma(v_j)) = -J \mathcal{E}(\sigma),
\end{equation*}
where $\delta(a,b)=1$ if $a=b$ and $0$ otherwise, and $\mathcal{E}(\sigma)$ is the number of edges whose endpoints are assigned the same spin under $\sigma$.
Remark that since $J < 0$ it is energetically favorable to have different spins at the endpoints of each edge, if possible.  This means 
that we are in the {\em antiferromagnetic} regime.

The Boltzmann distribution 
assigns a configuration $\sigma$ on $\Gamma$ probability proportional to the weight
\begin{equation*}
W_\Gamma(\sigma) = \exp \left(-\frac{H_\Gamma(\sigma)}{{\rm T}}\right) = \exp \left(\frac{J \mathcal{E}(\sigma)}{{\rm T}}\right),
\end{equation*}
where ${\rm T}>0$ is the temperature of the system \footnote{We set the
Boltzmann constant $k_B=1$.}.  The probability ${\rm Pr}(\sigma)$ of $\sigma$
occurring is therefore
\begin{equation}\label{DEF:PARTI}
{\rm Pr}(\sigma) = W_\Gamma(\sigma)/Z_\Gamma \quad \mbox{where} \quad Z_\Gamma:= \sum_\sigma W_\Gamma(\sigma),
\end{equation}
and the sum is over all possible spin configurations.
Some intuition for this distribution can be gained by considering the following two extreme cases:
when ${\rm T}$ is near zero, configurations with
minimum energy are strongly favored. Meanwhile for high temperature, all
configurations occur with nearly equal probability.

Let us introduce the temperature-like
variable $y:=e^{J/{\rm T}}$, so that
$W_\Gamma(\sigma)=y^{\mathcal{E}(\sigma)}.$ All the quantities above implicitly
depend on $q$, $y$, and the graph $\Gamma$. The normalizing factor $Z_\Gamma
(q, y)$ is called the \textit{partition function} and given by
\begin{align*}
Z_\Gamma(q,y):= \sum_\sigma y^{\mathcal{E}(\sigma)}.
\end{align*}
It turns out that $Z_\Gamma(q,y)$ is actually a polynomial in {\em both} $q$ and $y$.  To see this
it will be helpful to express the partition function in terms of $(q,v)$ where
$v = y-1$.   For any subset of the edge set
$A \subseteq E$ is a subgraph $(V, A)$.  We have
\begin{equation}\label{DEF:PARTI2}
 Z_\Gamma (q,v)  = \sum_\sigma \prod_{(i, j) \in E}\ [1 + v \delta(\sigma_i, \sigma_j)] = \sum_{A \subseteq E} q^{k(A)} v^{|A|}.
\end{equation}
where $k(A)$ is the number of connected components of $(V,A)$, including
isolated vertices.  This is called the {\em Fortuin-Kasteleyn \cite{FK} representation of
$Z_\Gamma (q,v)$}; see, for example, \cite[Section 2.2]{sokalsurvey}.  (We will
only express $Z_\Gamma$ in terms of $v$ instead of $y$ in this paragraph and in
Subsection \ref{SUBSEC:IRRED}.)

As discussed in the introduction, we will describe the zeros of $Z_\Gamma(q,y)$ as a divisor denoted
\begin{align*}
\mathcal{S}:=(Z_\Gamma(q, y)=0).
\end{align*}
Remark that in the next subsection we will see that if $\Gamma$ is $2$-connected, then $Z_\Gamma(q,y) = q \tilde{Z}_\Gamma(q,y)$
with $\tilde{Z}_\Gamma(q,y)$ irreducible, implying $\mathcal{S}$ is a {\em reduced divisor}, i.e.\ all multiplicities are one.
Therefore, if $\Gamma$ is $2$-connected
there is no harm in thinking of $\mathcal{S}$ as a (reducible) algebraic curve.

To establish the connection between the chromatic polynomial $\P_\Gamma(q)$ and the partition function $Z_\Gamma(q,y)$ of the Potts model note that
\begin{equation*}
\P_\Gamma(q)=\sum_{\substack{\sigma  \text{  such that} \\ \mathcal{E}(\sigma)=0}} 1 = Z_\Gamma(q,0).
\end{equation*}
Therefore, the chromatic zeros are given by the intersection:
\begin{equation*}\label{DEF:C}
\mathcal{C}:= \mathcal{S} \cap (y=0),
\end{equation*}
where Bezout intersection multiplicities and multiplicities of the divisor $\mathcal{S}$ are taken into account.

\subsection{Irreducibility of $\tilde{Z}_\Gamma(q,y)$ for $2$-connected $\Gamma$}
\label{SUBSEC:IRRED}
It follows from (\ref{DEF:PARTI2}) that we can always factor $Z_\Gamma (q,v) =
q \widetilde{Z}_\Gamma (q,v)$ in the polynomial ring $\mathbb{C}[q,v]$.
The goal of this subsection is to prove:
\begin{proposition}\label{PROP:IRRED}
  If $\Gamma$ is $2$-connected, then $\widetilde{Z}_\Gamma (q,v)$ is irreducible in $\mathbb{C}[q,v]$.  (The same holds in the $(q,y)$ variables.)
\end{proposition}

We will prove this proposition using the well-known relationship between $\widetilde{Z}_\Gamma (q,y)$ and the {\em Tutte Polynomial} of $\Gamma$.
It is defined as
\begin{equation}\label{DEFN:TUTTE}
\TUTTE_\Gamma (x,y) = \sum_{A \subset E} (x-1)^{k(A)-1} (y-1)^{|A|+k(A)-|V|},
\end{equation}
where $k(A)$ has the same interpretation as in (\ref{DEF:PARTI2}). 
The variables $(x,y)$ in the Tutte Polynomial are related\footnote{Although the variable $y$ appears
in Equation (\ref{DEF:PARTI}) for the partition function and also in Equation (\ref{DEFN:TUTTE}) for the Tutte Polynomial, there is no
conflict of notation because both satisfy $y=v+1$.}  to the
variables $(q,v)$ in the Partition Function (\ref{DEF:PARTI2})
by:
\begin{align*}
x = 1+(q/v) \qquad \mbox{and} \qquad y = v+1.
\end{align*}
Comparing (\ref{DEFN:TUTTE}) with (\ref{DEF:PARTI2}) we
see the following relationship \cite[Section 2.5]{sokalsurvey} between $\TUTTE(x,y)$ and $Z_\Gamma(q,v)$:
\begin{equation}\label{TUTTE-PARTI}
   \TUTTE_\Gamma (x,y) = (x-1)^{-1} (y-1)^{-|V|} Z_\Gamma ((x-1)(y-1), y-1).
\end{equation}

Proposition \ref{PROP:IRRED} will be a corollary to the following nice
result by de Mier, Merino, and Noyi \cite{merino}.

\begin{IRRETUTTE}
	If $\Gamma$ is a 2-connected graph, then $\TUTTE_\Gamma (x,y)$ is irreducible in $\mathbb{C}[x,y]$.
\end{IRRETUTTE}	

\noindent
Remark that the theorem proved in \cite{merino} is that the Tutte polynomial
of a connected matroid is irreducible.  However, that implies the result stated above because
associated to any graph $\Gamma$ is a matriod $M(\Gamma)$, called the {\em cycle matroid} of $\Gamma$, which has the following properties:
\begin{enumerate}
\item The Tutte polynomial of $\Gamma$ equals the Tutte polynomial of $M(\Gamma)$; (see e.g.
\cite[Equation 1.4]{sokalsurvey}),
\item $\Gamma$ is $2$-connected iff $M(\Gamma)$ is connected; (see e.g. \cite[p. 78]{WELSH}.)
\end{enumerate}

\begin{lemma}
        $Z_\Gamma(q,v)$ vanishes to order exactly $|V|$ at the origin.
\end{lemma}

\begin{proof}
 For any subgraph $(V, A)$, it follows from a counting argument that
 $k(A) + |A| \geq |V|$. Moreover, for the subgraph $(V, A_0)$ without any edges, the sum $k(A_0) + |A_0|$ is exactly $|V|$. Therefore the order of vanishing is exactly $|V|$ at the origin.
\end{proof}

\begin{proof}[Proof of Proposition \ref{PROP:IRRED}]
By the Irreducibility of the Tutte Polynomial, it suffices to prove that if $\widetilde{Z}_\Gamma$ is reducible then so is $\TUTTE_\Gamma$.
Suppose $\widetilde{Z}_\Gamma$ is reducible:
\begin{equation*}
\widetilde{Z}_\Gamma = A_1 \cdot A_2 \cdot B,
\end{equation*}
where $A_1, A_2$ are non-constant irreducible factors, and $B$ can potentially be a unit. Denote by $C_i$ the zero set of $A_i$.

Let $H:\mathbb{C}^2 \rightarrow \mathbb{C}^2$ be the birational map $(x,y)
\mapsto ((x-1)(y-1), y-1)$, so that by (\ref{TUTTE-PARTI}) we have
\begin{align*}
\TUTTE_\Gamma(x,y) = (y-1)^{-|V|+1} (\tilde{Z}_\Gamma \circ H ). 
\end{align*}
Therefore, in order to prove that $\TUTTE_\Gamma$ is reducible it suffices to find at least
two irreducible factors of $\tilde{Z}_\Gamma \circ H$ each of which is not equal to $y-1$.

For
$i=1$ and $2$, although $H^{-1} (C_i)$ can possibly contain the line $E:=\{(x,
y) \in \mathbb{C}^2: y=1\}$, it cannot be the only irreducible component of
$H^{-1}(C_i)$  because $H(E)$ is a single point $(0,0)$. From this observation
we now have to consider two separate cases.
\begin{enumerate}
        \item[\rm (i)] If $A_1 \not\equiv A_2$, then the zero set of $\TUTTE_\Gamma$ contains at least two distinct irreducible components, neither of which is the line $E$.
        \item[\rm (ii)] If $A_1 \equiv A_2$, then the zero set of $\TUTTE_\Gamma$ contains an irreducible component of multiplicity at least two, which is not the line $E$.
\end{enumerate}
In either case, we conclude that $\TUTTE_\Gamma$ is reducible.

\end{proof}

\subsection{Combinatorics of Hierarchical Lattices}

\begin{proposition}\label{PROP:2CONNECTED_GENERATING_GRAPH}
Suppose $\{\Gamma_n\}_{n=1}^\infty$ is a hierarchical lattice that is generated by a $2$-connected generated graph $\Gamma = (V,E)$.
Then, $\Gamma_n$ is $2$-connected for each $n \geq 0$.
\end{proposition}

\begin{proof}
The proof is by induction on $n$.   Since $\Gamma_0$ is a single edge with two vertices at its endpoints it is $2$-connected.  Suppose now that $\Gamma_n$ is $2$-connected for some $n \geq 0$ to show that $\Gamma_{n+1}$ is $2$-connected.
Recall that $\Gamma_{n+1}$ is built by replacing each edge of the generating graph $\Gamma$ with a copy of $\Gamma_n$
using the marked vertices $a$ and $b$ as endpoints.  The vertices of $\Gamma_{n+1}$ fall into two classes:
\begin{enumerate}
\item The $|V|$ vertices of $\Gamma_{n+1}$ that come from the vertices of $\Gamma$.  Each of them is a marked
vertex $a$ or $b$ from some copy of $\Gamma_n$, and
\item The remaining vertices.
\end{enumerate}
If the removal of a vertex of Type (1) disconnects $\Gamma_{n+1}$ then, since
each $\Gamma_n$ is $2$-connected, this would imply that removal of the
corresponding vertex of $\Gamma$ disconnects $\Gamma$.  This is impossible
because $\Gamma$ is $2$-connected.
Meanwhile, if removal of a vertex of Type (2) disconnects $\Gamma_{n+1}$ then its removal will also disconnect the unique copy
of $\Gamma_n$ that the vertex is contained in.  This contradicts the induction hypothesis.
\end{proof}

\begin{proposition}\label{vertexedge}
        Let $\Gamma_n = (V_n, E_n)$ be a hierarchical lattice generated by generating graph $\Gamma~=~(V,E)$. Then $|V_n|$ and $|E_n|$ grow at the same exponential rate as $n \rightarrow \infty$.
\end{proposition}

\begin{proof}
Observe that for any $n \geq 1$,
\begin{equation*}
|V_{n+1}| \ = \ |V_n| + |E_n| \cdot \left(|V|-2\right) \ = \ |V_n| + |E|^n \cdot \left(|V|-2\right).
\end{equation*}
It follows from induction that
\begin{equation*}
|V_n| \ = \ |V| + \left(|V|-2\right) \cdot \sum_{i=1}^{n-1}|E|^i \ = \ |V| + \left(|V|-2\right) \cdot |E|\frac{|E|^{n-1}-1}{|E|-1},
\end{equation*}
which proves the assertion.
\end{proof}

\subsection{Migdal-Kadanoff Renormalization for the DHL}
\label{SEC:MK_RENORM_DHL}

Let $\{\Gamma_n = (V_n,E_n)\}_{n=0}^\infty$ be the Diamond Hierarchical Lattice (DHL).  For each $n \geq 0$ the
partition function $Z_n(q,y) \equiv Z_{\Gamma_n}(q,y)$ has zero divisor
\begin{align*}
\mathcal{S}_n:=(Z_n(q, y)=0).
\end{align*}
Remark that $\Gamma_0$ is always a single edge with two vertices at its endpoints, so a simple calculation yields $Z_{\Gamma_n}(q,y) = q(y+q-1)$ so that
\begin{align*}
\mathcal{S}_0:= (q(y+q-1)=0).
\end{align*}
Associated to the hierarchical lattice $\{\Gamma_n\}_{n=0}^\infty$ is a Migdal-Kadanoff renormalization mapping that relates the zero divisor $\mathcal{S}_{n+1}$
to the zero divisor $\mathcal{S}_{n}$.

\begin{proposition}\label{DEFN:RENORM}
For the DHL we have that for each $n \geq 0$
\begin{align*}
\mathcal{S}_{n} = (R^{n})^*(\mathcal{S}_0)
\end{align*}
where $R:\mathbb{C} \times \mathbb{P}^1 \rightarrow \mathbb{C} \times \mathbb{P}^1$ is given by
\begin{equation}\label{EQN:DEF_RQ_FOR_DHL}
R(q,y)=\left(q, r_q(y) \right), \quad \mbox{where} \quad r_q(y)=\left(\frac{y^2+q-1}{2y+q-2}\right)^2.
\end{equation}
As usual, the superscript $*$ denotes pullback of a divisor and we will denote points $y \in \mathbb{P}^1$ using the standard chart $\mathbb{C} \subset \mathbb{P}^1$.
\end{proposition}

The proof will be very similar to the derivation of the Migdal-Kadanoff renormalization transformation for the Ising Model on the DHL \cite[Section 2.5]{BLR1}
and it relies on the multiplicativity of the conditional partition functions which is proved in \cite[Lemma 2.1]{BLR1}, in the context of the Ising Model.

\begin{proof}
For each $n \geq 0$ consider the following conditional partition functions:
\begin{equation*}
\U_n \equiv \U_n(q,y) := \sum_{\substack{\sigma \text{  such that} \\ \sigma(a)=\sigma(b)=1}} W(\sigma) \qquad \mbox{and} \qquad \V_n \equiv \V_n(q,y) :=\sum_{\substack{\sigma \text{  such that} \\ \sigma(a)=1, \,  \sigma(b)=2}} W(\sigma).
\end{equation*}
We claim for each $n\geq 0$ that
\begin{equation}\label{EQN:HOMOG_RECURSION}
\U_{n+1}=\left(\U_n^2+(q-1)\V_n^2\right)^2 \quad \mbox{and} \quad \V_{n+1}=\left(2\U_n\V_n+(q-2)\V_n^2\right)^2.
\end{equation}
To show this,
it will be helpful to depict them graphically as follows:
\begin{figure}[H]
	\centering
	\scalebox{0.7}{
		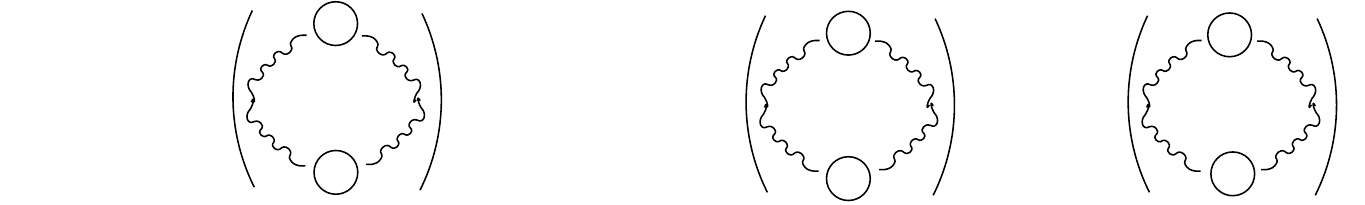
	}
\end{figure}
\noindent
The ones and twos in the figure denote the spins at the marked vertices $a$ and $b$.
Let us graphically illustrate the derivation of the first equation from (\ref{EQN:HOMOG_RECURSION}):
\begin{figure}[H]
    \centering
	\scalebox{0.7}{
		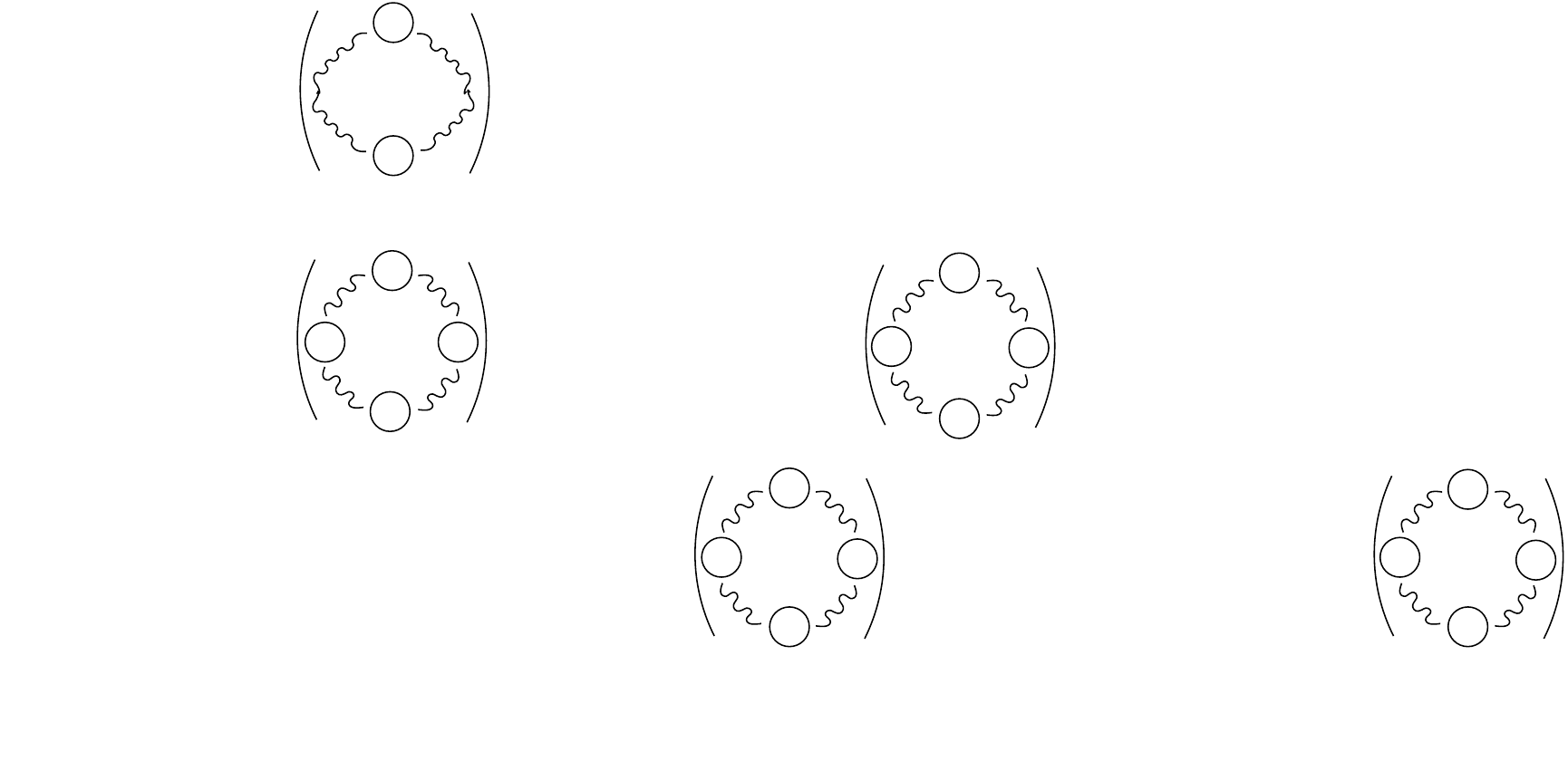
	}
\end{figure}
\noindent
The numbers one, two, and three in the second row of the figure above are meant to denote the boundary conditions
imposed on each of the four copies of $\Gamma_n$ that are glued together to
form $\Gamma_{n+1}$.  The third line is obtained from the second using multiplicativity of the conditional
partition functions.  (Once the spins at those four vertices are fixed, the conditional partition function
is the same as that of a disjoint union of the four copies of $\Gamma_n$, each with its corresponding
boundary conditions.) The expression for $\V_{n+1}$ in
(\ref{EQN:HOMOG_RECURSION}) can be obtained similarly. 

In order to use an iteration on $\mathbb{P}^1$ instead of $\mathbb{C}^2$ it will be 
more convenient to iterate the ratio $y_n := \U_n/\V_n$, where $n \geq 0$.  A simple calculation shows that
$y_0 = y = e^{J/T}$.  Using (\ref{EQN:HOMOG_RECURSION}) we find that 
\begin{equation*}
y_{n+1}=\frac{\U_{n+1}}{\V_{n+1}}
=\left(\frac{\U_n^2 + (q-1)\V_n^2}{2\U_n\V_n+(q-2)\V_n^2}\right)^2
=\left(\frac{y_n^2+q-1}{2y_n+q-2}\right)^2 = r_q(y_n).
\end{equation*}
Therefore, $(q_n,y_n) = R^n(q,y)$ where $q_n = q$ for all $n$.

Note that
\begin{equation}\label{EQN:DERIVATION_MK}
Z_n(q,y) = q \ \U_n+q(q-1)\V_n.
\end{equation}
Since the generating graph $\Gamma$ is $2$ connected Proposition \ref{PROP:2CONNECTED_GENERATING_GRAPH} implies that $\Gamma_n$ is $2$-connected 
for each $n \geq 0$.
Therefore, Proposition \ref{PROP:IRRED} gives that $\tilde{Z}_n(q,y) = \U_n+(q-1)\V_n$ is irreducible, implying
that $\U_n$ and $\V_n$ have no common factors of positive degree in $q$ or $y$.
Therefore,
\begin{align}\label{EQN:DERIVATION_MK2}
\mathcal{S}_{n} \, = \, (Z_n(q,y) = 0)\, = \, (q_n (\U_n+ (q_n-1)\V_n) = 0) \, = \, (q_n (y_n+q_n-1) = 0) \, = \, (R^n)^* \mathcal{S}_0,
\end{align}
where in the third equality we used that $\U_n$ and $\V_n$ have no common factors of positive degree.
\end{proof}


The map $r_q(y)$ given in (\ref{EQN:DEF_RQ_FOR_DHL}) is called the \textit{Migdal-Kadanoff renormalization mapping} for the $q$-state Potts model on the DHL.
Remark that this is an algebraic family of rational mappings of degree $4$ defined over $\mathbb{Q}$.
As a consequence of Proposition \ref{DEFN:RENORM}, the chromatic zeros for the DHL can be obtained dynamically:
\begin{align}\label{DYNZERO}
\mathcal{C}_n = (R^{n})^*(\mathcal{S}_0) \cap (y=0)
\end{align}
and note that up to the simple zero at $q=0$ we can use 
\begin{align}\label{DYNZERO1}
\tilde{\mathcal{C}}_n = (R^{n})^*(y+q-1=0) \cap (y=0).
\end{align}
When considering the limiting measure of chromatic zeros it suffices to consider $\tilde{\mathcal{C}}_n$.

\subsection{Migdal-Kadanoff Renormalization for arbitrary hierarchical lattices}
\label{SUBSEC:MK_ARBITRARY}

Now suppose $\Gamma_n = (V_n,E_n)$ is the hierarchical lattice generated by an
arbitrary generating graph $\Gamma = (V,E)$.  
It is clear that we can repeat the procedure in Proposition \ref{DEFN:RENORM}
to produce a renormalization mapping $r_q(y)$ associated to the generating
graph $\Gamma$, which is a rational map in $y$ on the Riemann sphere of degree
at most $|E|$, parameterized by polynomials in $q$ with integer coefficients.

However, it is possible that the generic degree of $r_q(y)$ is strictly
smaller than $|E|$. One such example is the Tripod shown in Figure \ref{generating_graphs}
for which we have
\begin{equation*}
\U_{n+1}= \left(\U_n+(q-1)\V_n\right) \left(\U_n^2+(q-1)\V_n^2\right) \quad \mbox{and} \quad \V_{n+1}=\left(\U_n+(q-1)\V_n\right) \left(2\U_n\V_n+(q-2)\V_n^2\right).
\end{equation*}
The common factor of positive degree $\left(\U_n+(q-1)\V_n\right)$ is a consequence of the ``horizontal'' edge that
is connected to the remainder of the generating graph $\Gamma$ at a single vertex.   When taking the ratios 
$y_n = \U_n / \V_n$ we lose track of these common factors resulting in the drop of generic degree:
\begin{equation*}
  R(q,y) = (q,r_q(y)) \quad \mbox{where} \quad      r_q(y)= \frac{y^2+q-1}{2y+q-2},
\end{equation*}
which has degree two even though $\Gamma$ has three edges. 
This drop in generic degree results in 
$(R^n)^* \mathcal{S}_0 < (Z_n(q,y))$ for the hierarchical lattice
generated by the Tripod. 

This phenomenon can be avoided if the generating graph is $2$-connected and the proof is exactly the
same as for the DHL.  We summarize:
\begin{proposition}\label{NON-DEGENERATE-DEGREE}
	Let $\{\Gamma_n\}_{n=0}^\infty$ be the hierarchical lattice generated by $\Gamma = (V,E)$. If $\Gamma$ is $2$-connected, then the associated renormalization mapping $R(q,y) = (q,r_q(y))$ has generic degree $|E|$ and satisfies
\begin{align*}
\mathcal{S}_{n} = (R^{n})^*(\mathcal{S}_0),
\end{align*}
where $\mathcal{S}_n = (Z_n(q,y))$ and $\mathcal{S}_0 = (q(y+q-1))$.  Moreover, $r_q$ is defined over $\mathbb{Q}$.
\end{proposition}

\noindent
Several concrete examples are presented in Section \ref{SEC:EXAMPLES}.

\section{Proof of Theorem A}\label{SEC:PROOF_THMA}

Let $\{\Gamma_n\}_{n=1}^\infty$ be a hierarchical lattice whose generating
graph $\Gamma=(V,E)$ is $2$-connected. Denote its Migdal-Kadanoff
renormalization mapping by $R(q,y) = (q,r_q(y))$. Since 
$\Gamma$ is $2$-conneced, Proposition \ref{NON-DEGENERATE-DEGREE} implies that the chromatic zeros
for $\Gamma_n$ (omitting the simple zero at $q=0$) are given by
$\tilde{\mathcal{C}}_n = (R^{n})^*(y+q-1=0) \cap (y=0)$.  Therefore, in the
language of currents,
\begin{align*}
\tilde{\mu}_n \, := \, \frac{1}{|V_n|} \sum_{\substack{q \in \mathbb{C} \setminus \{0\} \\ \P_{\Gamma_n}(q)=0}} \delta_q \, = \, (\pi_{1})_* \left(\frac{1}{|V_n|} (R^{n})^* [y+q-1 =0] \ \wedge \ [y=0]\right),
\end{align*}
where the zeros of $\P_{\Gamma_n}(q)$ are counted with multiplicities, as always.
Since $\tilde{\mu}_n$ and $\mu_n$ (see (\ref{limitingmeasure})) differ by $1/|V_n|$ times a Dirac measure at $q=0$, it suffices
to prove that the sequence $\tilde{\mu}_n$ converges.  Moreover, Proposition \ref{vertexedge} allows
us to replace the normalizing factor of $|V_n|$ with $|E_n|$.  Therefore, it suffices to verify
that $R = (q,r_q(y))$ and the marked points $a(q)=0$ and $b(q)=1-q$ satisfy the hypotheses
of Theorem C'. 

By Proposition \ref{NON-DEGENERATE-DEGREE}, the algebraic family $r_q$ is
defined over $\mathbb{Q}$.  Hypotheses (i) and (ii) on the marked points will
be verified in Propositions \ref{noiterate} and \ref{nonexceptional} below.

\begin{proposition}\label{noiterate}
There are no iterates $n \geq 0$ satisfying $r_q^n (0) \equiv 1-q$.
\end{proposition}

\begin{proof}
Away from the finitely many points in $\VDEG$, the 
chromatic zeros of $\Gamma_n$ are solutions in~$q$ to $r_q^n(0)=1-q$.  If there is some iterate $n \geq 0$ such that $r_q^n (0) \equiv 1-q$, this will imply that $\Gamma_n$ has infinitely many chromatic zeros, which is impossible
because ${\rm deg}(\P_{\Gamma_n}) = |V_n|$.
\end{proof}



\begin{proposition}\label{nonexceptional}
The marked point $b(q)=1-q$ is not persistently exceptional for $r_q$.
\end{proposition}

\begin{proof}
Assume by contradiction that the marked point $b(q)=1-q$ is persistently
exceptional.  Taking the second iterate, we can suppose it is a fixed point.
Then by (\ref{EQN:DERIVATION_MK2}), the pullback of the divisor $(y=1-q)$ by
the map $R^2$ satisfies
\begin{align*}
(\tilde{Z}_2(q,y))=\left(R^2\right)^* (y=1-q) = |E|^2 (y=1-q),
\end{align*}
which implies that the partition function, $\tilde{Z}_2(q,y)=(y+q-1)^{|E|^2}$, for $\Gamma_2$ is reducible. However, since the generating graph $\Gamma$ is assumed to be $2$-connected, $\Gamma_2$ is also $2$-connected, so $\tilde{Z}_2(q,y)$ is irreducible by Proposition \ref{PROP:IRRED}, which is a contradiction.




\end{proof}

\qed (Theorem A)

 



\section{Proof of Theorem B}\label{SEC:PROOF_THMB}

This is the only section of the paper where we will use marked points that are critical.
We will use the following famous result which appears as Corollary 1.6 from \cite{McM1}:

\begin{theorem}[\bf McMullen \cite{McM1}]\label{McMullen1}
For any holomorphic family of rational maps over the unit disk $\Delta$, the bifurcation locus $B(f) \subset \Delta$ is either empty or has Hausdorff dimension two.
\end{theorem}

Although the above theorem states that the bifurcation locus, which is the
union of the active loci of all the critical points, has Hausdorff dimension
two (unless it is empty), one can check that the proof still applies to each
individual marked critical point $c(\lambda)$, as long as it bifurcates.
Indeed the proof of Theorem \ref{McMullen1} consists of using activity of the
marked point to construct a holomorphically-varying family of polynomial-like
mappings,  whose critical point is the marked one $c(\lambda)$.  Associated to
this family is the space of parameters $\lambda$ for which the orbit of the
critical point remains bounded (in the polynomial-like mapping).  McMullen
shows that this set is a quasiconformal image of the Mandelbrot set (or a
higher degree generalization).  The boundary of this ``baby Mandelbrot set''
has Hausdorff Dimension two \cite{Shi1}, and, by definition, the marked point
$c(\lambda)$ is active at such points.

\begin{proof}[Proof of Theorem B]
Using an analogous proof to that of Proposition \ref{DEFN:RENORM} one finds that the renormalization mapping for the $k$-fold DHL is
\begin{equation}\label{EQN:DHL_RQ}
	r_q(y)=\left(\frac{y^2+q-1}{2y+q-2}\right)^k.
\end{equation}
For this family of mappings we have $\VDEG = \{0,\infty\}$.  Since the generating graph is $2$-connected
Theorem A implies that the limiting measure of chromatic zeros exists for the $k$-fold DHL and the proof
of Theorem A implies that on $\mathbb{C} \setminus \VDEG$ it coincides with the activity measure
for the marked point $a(q) \equiv 0$.

One can check that $c(q):=\sqrt{1-q}$ is a critical point for $r_q(y)$, which
we can suppose is marked after replacing $\mathbb{C}$ with a branched cover.  A
direct calculation shows that $r_q(c(q)) \equiv 0 \equiv a(q)$.
Therefore, the activity loci of marked point $a(q)$ (and hence of our limiting measure
of $\mu$ of chromatic zeros) coincides with the activity locus for the critical point  $c(q)$.

It remains to check that these are non-empty and not entirely contained in
the set of parameters for which the degree of $r_q(y)$ drops.
Drop in degree of $r_q(y)$ corresponds to values of $q$ for which numerator and denominator of $r_q(y)$ have
a common zero.  One can check that this only happens when $q=0$.

One can also check by direct calculation that $y=1$ and $y=\infty$ are both
persistently superattracting fixed points for $r_q(y)$.  One has that $r_q(0)$
is a degree $k \geq 2$ rational function of $q$ and that $r_0(0) = (1/2)^k$.
Therefore, there is some parameter $q_1 \neq 0$ for which $r_{q_1}(0) = 1$.
On some open neighborhood of this parameter one has $r_q^n(0) \rightarrow 1$.
Meanwhile, one has $r_2(0) = \infty$ and so there is an open neighborhood of
$q=2$ on which $r_q^n(0) \rightarrow \infty$.  This implies that the marked
point $a(q)$ cannot be passive on the connected set $\mathbb{C} \setminus \{0\}$
by the identity theorem.

Theorem \ref{McMullen1} and the paragraph following it then give that the activity locus
of $c(q)$ has Hausdorff Dimension equal to two. 

\end{proof}

In the special case that $k = 2$, Laura DeMarco and Niki Myrto Mavraki observed the following:
\begin{proposition}
\label{PROP:SYMMETRY}
Let $r_q(y)$ be the renormalization mapping for the $2$-fold DHL given by (\ref{EQN:DHL_RQ}) with $k=2$.
Then, $B(r_q) = {\rm supp}(T_a)$.
\end{proposition}

\begin{proof}
The critical points of the map $r_q(y)$
are $y=1, 1-q, \infty, \frac{2-q}{2}, \pm \sqrt{1-q}$. Three of them behave similarly: $y=1$ and $y=\infty$ are superattracting fixed points, while $y=\frac{2-q}{2}$ is just a preimage of $\infty$. Meanwhile, note that $\pm \sqrt{1-q}$ are both preimages of $y=0$, so the bifurcation locus of the family is the union of the activity loci of the two marked points $y=0$, $y=1-q$.

The map $r_q$ commutes with 
\begin{align*}
	C_q (y):=\left( \frac{y+q-1}{y-1} \right)^2,
\end{align*}	
which satisfies $C_q(1-q)=0$ and $C_q(0) = r_q (1-q)=(1-q)^2$. Therefore, the activity loci of $y=0$ and $y=1-q$ coincide, and it follows that the bifurcation locus of the family is equal to the activity locus of the non-critical marked point $y=0$.
\end{proof}

\section{Examples}
\label{SEC:EXAMPLES}
We conclude the paper with a discussion of the chromatic zeros associated with the hierarchical lattices generated by the  
graphs shown in Figure \ref{generating_graphs}.  We also provide a more detailed
explanation of Figures \ref{FIG:DHLSUPP} and \ref{FIG:SPLIT_DIAMOND}.

\subsection{Linear Chain}
In this case, each graph $\Gamma_n$ is a tree so that
$\P_{\Gamma_n}(q) = q(q-1)^{|V_n|}$.  See, for example, \cite{harary}.  Therefore, the limiting measure of chromatic zeros for the
linear chain is a Dirac measure at $q=1$.

Meanwhile, even though the generating graph is not $2$-connected, the statement of Proposition~\ref{NON-DEGENERATE-DEGREE}
still applies with 
\begin{align*}
r_q(y) = \frac{y^2+q-1}{2y+q-2},
\end{align*}
which is the same formula as for the $k$-fold DHL, except with exponent $k=1$.  One can check that $r_q$ 
has $y=1-q$ as a persistent exceptional point, so that Theorem C' does not apply.  Indeed, the activity locus
for marked point $a(q) \equiv 0$ is the round circle $|q-1/2| = 1/2$ while for each $n \geq 0$ the sequence of wedge products (\ref{EQN:THMCPRIME_CONV}) 
is just the Dirac measure at $q=1$.

\subsection{$k$-fold DHL, where $k \geq 2$}
In the proofs of 
Theorems A and B we already saw that the limiting measure
$\mu$ of chromatic zeros exists for this lattice and that outside of $\VDEG =
\{q=0,\infty\}$ it coincides with the activity measure for the marked point
$a(q) \equiv 0$.  Here, we will explain the claim the activity locus, and hence
${\rm supp}(\mu)$, is the boundary between any two of the colors (blue, black,
and white) in Figure \ref{FIG:DHLSUPP}.  

The Migdal-Kadanoff renormalization mapping is given by (\ref{EQN:DHL_RQ}).
One can check that this mapping has $y=1$ and $y=\infty$ as persistent
superattracting fixed points.  In Figure \ref{FIG:DHLSUPP}, the set $q$ for
which $r_q^n(0) \rightarrow 1$ is shown in white (i.e.\ not colored) and the
set of $q$ for which $r_q^n(0) \rightarrow \infty$ is shown in blue.  Each of
these corresponds to passive behavior for the marked point $a(q) \equiv 0$.
Meanwhile, if there is some neighborhood $N$ of $q_0 \in \mathbb{C} \setminus
\VDEG$ on which $r_q^n(0)$ does not have one of these two behaviors, then
Montel's Theorem implies that $a(q)$ is also passive on $N$.  Such points are
colored black.

Conversely, if $q_0$ is on the boundary of two colors (blue, black, and white), then $q_0$ is an active 
parameter for the marked point $a(q)$.  Indeed, if $N$ is any neighborhood of $q_0$ then along any subsequence
$n_k$ we have that $r_q^{n_k}(0)$ will converge uniformly to $1$ or $\infty$ the parts of $N$ that are white 
or blue, respectively, and $r_q^{n_k}(0)$ will remain bounded away from $1$ and $\infty$ on the black.
Therefore, $r_q^n(0)$ cannot form a normal family on $N$.

\subsection{Triangles}
\label{SUBSEC:TRIANGLES}
As the generating graph is $2$-connected, Proposition~\ref{NON-DEGENERATE-DEGREE} applies and one can compute
that the Migdal-Kadanoff renormalization mapping is:
\begin{align}\label{EQN:RENORM_TRIANGLES}
r_q(y) = y\left(\frac{y^2+q-1}{2y+q-2}\right).
\end{align}
It is the same as for the linear chain, but with an extra factor of $y$.
Notice that for this family of mappings $\VDEG = \{q=0,2,\infty\}$.  The proof
of Theorem A applies and one concludes that on $\mathbb{C} \setminus \VDEG$ the
limiting measure of chromatic zeros $\mu$ coincides with the activity measure
of the marked point $a(q) \equiv 0$.  However, a curious thing happens: for
every iterate $n$ we have $r_q^n(0) = 0$ so that the marked point $a$ is
globally passive on $\mathbb{C} \setminus \VDEG$.  Therefore, $\mu$ is
supported on $\VDEG$.  This illustrates why it was important to use Theorem C'
(instead of just Theorem C) when proving Theorem A.  Working inductively with
(\ref{EQN:RENORM_TRIANGLES}) one can directly prove that $\mu$ is the Dirac
measure at $q=2$.

\subsection{Tripods}
As explained in Section \ref{SUBSEC:MK_ARBITRARY}, the Migdal-Kadanoff renormalization mapping for
the tripod coincides with that of the linear chain, due to a common factor appearing in the numerator and denominator.
This drop in degree makes $r_q$ not useful for studying the chromatic zeros on the hierarchical lattice
generated by the tripod.
However, since each of the graphs $\Gamma_n$ in this hierarchical lattice
is a tree, the limiting measure of chromatic zeros exists and is a Dirac measure at $q=1$,
by the same reasoning as for the linear chain.

\subsection{Split Diamonds}
The split diamond is $2$-connected and Theorem A implies that there is a limiting measure of chromatic
zeros $\mu$ for the associated lattice.  One can check that the Migdal-Kadanoff renormalization mapping 
for this generating graph is
\begin{align}\label{EQN:RENORM_SPLIT_DIAMOND}
r_q(y) = \frac{y^5+2(q-1)y^2+(q-1)y+(q-1)(q-2)}{2y^3+2y^2+5(q-2)y+(q-2)(q-3)}.
\end{align}
As for the $k$-fold DHL, one can check that $r_q$ has $y=1$ and $y=\infty$ as persistent superattracting fixed points.
Therefore, one can use the the same coloring scheme as for the $k$-fold DHL to make computer images
of the activity locus of $a(q) \equiv 0$, and hence of ${\rm supp}(\mu)$; See Figure \ref{FIG:SPLIT_DIAMOND}.
With some explicit calculations, one can rigorously verify that each of the three behaviors (white, blue,
and black) actually occurs for $q \not \in \VDEG$.

\bibliographystyle{plain}
\bibliography{biblio}

\begin{thebibliography}{10}

\bibitem{FRACTAL}
Fractalstream dynamical systems software.
\newblock Written by Matthew Noonan.
  \url{https://code.google.com/archive/p/fractalstream/}.

\bibitem{PhysRevB.29.2630}
David Andelman and A.~Nihat Berker.
\newblock Scale-invariant quenched disorder and its stability criterion at
  random critical points.
\newblock {\em Phys. Rev. B}, 29:2630--2635, Mar 1984.

\bibitem{MR2480740}
Magnus Aspenberg and Michael Yampolsky.
\newblock Mating non-renormalizable quadratic polynomials.
\newblock {\em Comm. Math. Phys.}, 287(1):1--40, 2009.

\bibitem{BAXTER}
Rodney~J. Baxter.
\newblock {\em Exactly solved models in statistical mechanics}.
\newblock Academic Press, Inc. [Harcourt Brace Jovanovich, Publishers], London,
  1989.
\newblock Reprint of the 1982 original.

\bibitem{shrock}
Laura Beaudin, Joanna Ellis-Monaghan, Greta Pangborn, and Robert Shrock.
\newblock A little statistical mechanics for the graph theorist.
\newblock {\em Discrete Math.}, 310(13-14):2037--2053, 2010.

\bibitem{BKW}
S.~Beraha, J.~Kahane, and N.~J. Weiss.
\newblock Limits of zeroes of recursively defined polynomials.
\newblock {\em Proc. Nat. Acad. Sci. U.S.A.}, 72(11):4209, 1975.

\bibitem{BO}
A~N Berker and S~Ostlund.
\newblock Renormalisation-group calculations of finite systems: order parameter
  and specific heat for epitaxial ordering.
\newblock {\em Journal of Physics C: Solid State Physics}, 12(22):4961--4975,
  nov 1979.

\bibitem{Bert1}
Fran\c{c}ois Berteloot.
\newblock Bifurcation currents in holomorphic families of rational maps.
\newblock In {\em Pluripotential theory}, volume 2075 of {\em Lecture Notes in
  Math.}, pages 1--93. Springer, Heidelberg, 2013.

\bibitem{BDS}
N.~L. Biggs, R.~M. Damerell, and D.~A. Sands.
\newblock Recursive families of graphs.
\newblock {\em J. Combinatorial Theory Ser. B}, 12:123--131, 1972.

\bibitem{BS}
Norman Biggs and Robert Shrock.
\newblock {$T=0$} partition functions for {P}otts antiferromagnets on square
  lattice strips with (twisted) periodic boundary conditions.
\newblock {\em J. Phys. A}, 32(46):L489--L493, 1999.

\bibitem{birkhoff}
G.~D. Birkhoff and D.~C. Lewis.
\newblock Chromatic polynomials.
\newblock {\em Trans. Amer. Math. Soc.}, 60:355--451, 1946.

\bibitem{birkhoff1}
George~D. Birkhoff.
\newblock A determinant formula for the number of ways of coloring a map.
\newblock {\em Ann. of Math. (2)}, 14(1-4):42--46, 1912/13.

\bibitem{BLR1}
P.~Bleher, M.~Lyubich, and R.~Roeder.
\newblock Lee--{Y}ang zeros for the {DHL} and 2{D} rational dynamics, {I}.
  {F}oliation of the physical cylinder.
\newblock {\em J. Math. Pures Appl. (9)}, 107(5):491--590, 2017.

\bibitem{BL}
P.~M. Bleher and M.~Yu. Lyubich.
\newblock Julia sets and complex singularities in hierarchical {I}sing models.
\newblock {\em Comm. Math. Phys.}, 141(3):453--474, 1991.

\bibitem{BZ1}
P.~M. Bleher and E.~\v{Z}alys.
\newblock Existence of long-range order in the {M}igdal recursion equations.
\newblock {\em Comm. Math. Phys.}, 67(1):17--42, 1979.

\bibitem{BZ3}
P.~M. Bleher and E.~\v{Z}alys.
\newblock Asymptotics of the susceptibility for the {I}sing model on the
  hierarchical lattices.
\newblock {\em Comm. Math. Phys.}, 120(3):409--436, 1989.

\bibitem{BLR2}
Pavel Bleher, Mikhail Lyubich, and Roland Roeder.
\newblock Lee-{Y}ang-{F}isher zeros for the {DHL} and 2{D} rational dynamics,
  {II}. {G}lobal pluripotential interpretation.
\newblock {\em J. Geom. Anal.}, 30(1):777--833, 2020.

\bibitem{BZ2}
P.~M. Blekher and \`E. Zhalis.
\newblock Limit {G}ibbs distributions for the {I}sing model on hierarchical
  lattices.
\newblock {\em Litovsk. Mat. Sb.}, 28(2):252--268, 1988.

\bibitem{callsilverman}
Gregory~S. Call and Joseph~H. Silverman.
\newblock Canonical heights on varieties with morphisms.
\newblock {\em Compositio Math.}, 89(2):163--205, 1993.

\bibitem{MR4118578}
Shu-Chiuan Chang, Roland K.~W. Roeder, and Robert Shrock.
\newblock {$q$}-plane zeros of the {P}otts partition function on diamond
  hierarchical graphs.
\newblock {\em J. Math. Phys.}, 61(7):073301, 32, 2020.

\bibitem{shrock2}
Shu-Chiuan Chang and Robert Shrock.
\newblock Tutte polynomials and related asymptotic limiting functions for
  recursive families of graphs.
\newblock {\em Adv. in Appl. Math.}, 32(1-2):44--87, 2004.
\newblock Special issue on the Tutte polynomial.

\bibitem{CS}
Shu-Chiuan Chang and Robert Shrock.
\newblock Zeros of the {P}otts model partition function on {S}ierpinski graphs.
\newblock {\em Phys. Lett. A}, 377(9):671--675, 2013.

\bibitem{CHJR}
Ivan Chio, Caleb He, Anthony~L. Ji, and Roland K.~W. Roeder.
\newblock Limiting measure of {L}ee-{Y}ang zeros for the {C}ayley tree.
\newblock {\em Comm. Math. Phys.}, 370(3):925--957, 2019.

\bibitem{DEMAILLY}
Jean-Pierre Demailly.
\newblock Monge-{A}mp\`ere operators, {L}elong numbers and intersection theory.
\newblock In {\em Complex analysis and geometry}, Univ. Ser. Math., pages
  115--193. Plenum, New York, 1993.

\bibitem{Dem1}
Laura DeMarco.
\newblock Dynamics of rational maps: a current on the bifurcation locus.
\newblock {\em Math. Res. Lett.}, 8(1-2):57--66, 2001.

\bibitem{DeM2}
Laura DeMarco.
\newblock Dynamics of rational maps: {L}yapunov exponents, bifurcations, and
  capacity.
\newblock {\em Math. Ann.}, 326(1):43--73, 2003.

\bibitem{DeM3}
Laura DeMarco.
\newblock Bifurcations, intersections, and heights.
\newblock {\em Algebra Number Theory}, 10(5):1031--1056, 2016.

\bibitem{DEMARCO_KAWA}
Laura DeMarco.
\newblock Dynamical moduli spaces and elliptic curves.
\newblock {\em Ann. Fac. Sci. Toulouse Math. (6)}, 27(2):389--420, 2018.

\bibitem{DDI}
B.~Derrida, L.~de~Seze, C.~Itzykson, and and.
\newblock Fractal structure of zeros in hierarchical models.
\newblock {\em J. Statist. Phys.}, 33(3):559--569, 1983.

\bibitem{DS1}
Tien-Cuong Dinh and Nessim Sibony.
\newblock Dynamics in several complex variables: endomorphisms of projective
  spaces and polynomial-like mappings.
\newblock In {\em Holomorphic dynamical systems}, volume 1998 of {\em Lecture
  Notes in Math.}, pages 165--294. Springer, Berlin, 2010.

\bibitem{DKT_BOOK}
F.~M. Dong, K.~M. Koh, and K.~L. Teo.
\newblock {\em Chromatic polynomials and chromaticity of graphs}.
\newblock World Scientific Publishing Co. Pte. Ltd., Hackensack, NJ, 2005.

\bibitem{Duj1}
Romain Dujardin.
\newblock Bifurcation currents and equidistribution in parameter space.
\newblock In {\em Frontiers in complex dynamics}, volume~51 of {\em Princeton
  Math. Ser.}, pages 515--566. Princeton Univ. Press, Princeton, NJ, 2014.

\bibitem{DF}
Romain Dujardin and Charles Favre.
\newblock Distribution of rational maps with a preperiodic critical point.
\newblock {\em Amer. J. Math.}, 130(4):979--1032, 2008.

\bibitem{FAVRE_GAUTHIER}
Charles Favre and Thomas Gauthier.
\newblock The arithmetic of polynomial dynamical pairs.
\newblock Preprint: https://arxiv.org/abs/2004.13801.

\bibitem{FK}
C.~M. Fortuin and P.~W. Kasteleyn.
\newblock On the random-cluster model. {I}. {I}ntroduction and relation to
  other models.
\newblock {\em Physica}, 57:536--564, 1972.

\bibitem{FLM}
Alexandre Freire, Artur Lopez, and Ricardo Ma\~{n}\'{e}.
\newblock An invariant measure for rational maps.
\newblock {\em Bol. Soc. Brasil. Mat.}, 14(1):45--62, 1983.

\bibitem{GAUTHIER1}
Thomas Gauthier.
\newblock Dynamical pairs with an absolutely continuous bifurcation measure.
\newblock Preprint: https://arxiv.org/abs/1810.02385.

\bibitem{GV}
Thomas Gauthier and Gabriel Vigny.
\newblock Distribution of points with prescribed derivative in polynomial
  dynamics.
\newblock {\em Riv. Math. Univ. Parma (N.S.)}, 8(2):247--270, 2017.

\bibitem{GH}
Phillip Griffiths and Joseph Harris.
\newblock {\em Principles of algebraic geometry}.
\newblock Wiley-Interscience [John Wiley \& Sons], New York, 1978.
\newblock Pure and Applied Mathematics.

\bibitem{GK}
Robert~B. Griffiths and Miron Kaufman.
\newblock Spin systems on hierarchical lattices. introduction and thermodynamic
  limit.
\newblock {\em Phys. Rev. B}, 26:5022--5032, Nov 1982.

\bibitem{harary}
Frank Harary.
\newblock {\em Graph theory}.
\newblock Addison-Wesley Publishing Co., Reading, Mass.-Menlo Park,
  Calif.-London, 1969.

\bibitem{hormander}
Lars H\"ormander.
\newblock {\em The analysis of linear partial differential operators. {I}}.
\newblock Springer Study Edition. Springer-Verlag, Berlin, second edition,
  1990.
\newblock Distribution theory and Fourier analysis.

\bibitem{PhysRevB.80.134201}
Ferenc Igl\'oi and Lo\"{\i}c Turban.
\newblock Disordered potts model on the diamond hierarchical lattice:
  Numerically exact treatment in the large-$q$ limit.
\newblock {\em Phys. Rev. B}, 80:134201, Oct 2009.

\bibitem{JACKSON_CONJ}
Bill Jackson.
\newblock Zeros of chromatic and flow polynomials of graphs.
\newblock {\em J. Geom.}, 76(1-2):95--109, 2003.
\newblock Combinatorics, 2002 (Maratea).

\bibitem{JPS}
Bill Jackson, Aldo Procacci, and Alan~D. Sokal.
\newblock Complex zero-free regions at large {$|q|$} for multivariate {T}utte
  polynomials (alias {P}otts-model partition functions) with general complex
  edge weights.
\newblock {\em J. Combin. Theory Ser. B}, 103(1):21--45, 2013.

\bibitem{JLUO}
Luo Jiaqi.
\newblock Combinatorics and holomorphic dynamics: Captures, matings, newton's
  method, 1995.
\newblock Thesis, Cornell University.

\bibitem{KADANOFF1976}
L.~Kadanoff.
\newblock Notes on migdal's recursion formulas.
\newblock {\em Annals of Physics}, 100:359--394, 1976.

\bibitem{PhysRevB.23.3421}
Wolfgang Kinzel and Eytan Domany.
\newblock Critical properties of random potts models.
\newblock {\em Phys. Rev. B}, 23:3421--3434, Apr 1981.

\bibitem{LEP}
M.~Ju. Ljubich.
\newblock Entropy properties of rational endomorphisms of the {R}iemann sphere.
\newblock {\em Ergodic Theory Dynam. Systems}, 3(3):351--385, 1983.

\bibitem{Ly2}
M.~Yu. Lyubich.
\newblock Some typical properties of the dynamics of rational mappings.
\newblock {\em Uspekhi Mat. Nauk}, 38(5(233)):197--198, 1983.

\bibitem{MR751394}
M.~Yu. Lyubich.
\newblock Investigation of the stability of the dynamics of rational functions.
\newblock {\em Teor. Funktsi\u{\i} Funktsional. Anal. i Prilozhen.},
  (42):72--91, 1984.
\newblock Translated in Selecta Math. Soviet. {{\bf{9}}} (1990), no. 1, 69--90.

\bibitem{MR}
Colin Maclachlan and Alan~W. Reid.
\newblock {\em The arithmetic of hyperbolic 3-manifolds}, volume 219 of {\em
  Graduate Texts in Mathematics}.
\newblock Springer-Verlag, New York, 2003.

\bibitem{McM1}
Curtis~T. McMullen.
\newblock The {M}andelbrot set is universal.
\newblock In {\em The {M}andelbrot set, theme and variations}, volume 274 of
  {\em London Math. Soc. Lecture Note Ser.}, pages 1--17. Cambridge Univ.
  Press, Cambridge, 2000.

\bibitem{merino}
C.~Merino, A.~de~Mier, and M.~Noy.
\newblock Irreducibility of the {T}utte polynomial of a connected matroid.
\newblock {\em J. Combin. Theory Ser. B}, 83(2):298--304, 2001.

\bibitem{MIGDAL1}
A.~A. Migdal.
\newblock Phase transitions in gauge and spin-lattice systems.
\newblock {\em JETP}, 69:1457--1467, 1975.

\bibitem{MIGDAL2}
A.~A. Migdal.
\newblock Recurrence equations in gauge field theory.
\newblock {\em JETP}, 69:810--822, 1975.

\bibitem{okuyama}
Y\^usuke Okuyama.
\newblock Equidistribution of rational functions having a superattracting
  periodic point towards the activity current and the bifurcation current.
\newblock {\em Conform. Geom. Dyn.}, 18:217--228, 2014.

\bibitem{PR2}
H.~Peters and G.~Regts.
\newblock Location of zeros for the partition function of the {I}sing model on
  bounded degree graphs.
\newblock Preprint: \url{https://arxiv.org/abs/1810.01699}.

\bibitem{PR1}
H.~Peters and G.~Regts.
\newblock On a conjecture of sokal concerning roots of the independence
  polynomial.
\newblock {\em Michigan Mathematical Journal}, 68(1):33--55, 2019.

\bibitem{ROYLESOKAL}
G.~F. Royle and A.~Sokal.
\newblock Linear bound in terms of maxmaxflow for the chromatic roots of
  series-parallel graphs.
\newblock arXiv version, see \url{https://arxiv.org/abs/1307.1721}.

\bibitem{SS1}
Jes\'{u}s Salas and Alan~D. Sokal.
\newblock Transfer matrices and partition-function zeros for antiferromagnetic
  {P}otts models. {I}. {G}eneral theory and square-lattice chromatic
  polynomial.
\newblock {\em J. Statist. Phys.}, 104(3-4):609--699, 2001.

\bibitem{SS6}
Jes\'{u}s Salas and Alan~D. Sokal.
\newblock Transfer matrices and partition-function zeros for antiferromagnetic
  {P}otts models {VI}. {S}quare lattice with extra-vertex boundary conditions.
\newblock {\em J. Stat. Phys.}, 144(5):1028--1122, 2011.

\bibitem{Shi1}
Mitsuhiro Shishikura.
\newblock The boundary of the {M}andelbrot set has {H}ausdorff dimension two.
\newblock {\em Ast\'{e}risque}, (222):7, 389--405, 1994.
\newblock Complex analytic methods in dynamical systems (Rio de Janeiro, 1992).

\bibitem{shrock3}
Robert Shrock.
\newblock Chromatic polynomials and their zeros and asymptotic limits for
  families of graphs.
\newblock {\em Discrete Math.}, 231(1-3):421--446, 2001.
\newblock 17th British Combinatorial Conference (Canterbury, 1999).

\bibitem{ST97}
Robert Shrock and Shan-Ho Tsai.
\newblock Asymptotic limits and zeros of chromatic polynomials and ground-state
  entropy of potts antiferromagnets.
\newblock {\em Phys. Rev. E}, 55:5165--5178, May 1997.

\bibitem{ST}
Robert Shrock and Shan-Ho Tsai.
\newblock Ground-state entropy of the {P}otts antiferromagnet on cyclic strip
  graphs.
\newblock {\em J. Phys. A}, 32(17):L195--L200, 1999.

\bibitem{SIBONYDYNAMICS}
Nessim Sibony.
\newblock Dynamique des applications rationnelles de {$\bold P^k$}.
\newblock In {\em Dynamique et g\'{e}om\'{e}trie complexes ({L}yon, 1997)},
  volume~8 of {\em Panor. Synth\`eses}, pages ix--x, xi--xii, 97--185. Soc.
  Math. France, Paris, 1999.

\bibitem{silverman}
Joseph~H. Silverman.
\newblock Integer points, {D}iophantine approximation, and iteration of
  rational maps.
\newblock {\em Duke Math. J.}, 71(3):793--829, 1993.

\bibitem{silvermanbook}
Joseph~H. Silverman.
\newblock {\em The arithmetic of dynamical systems}, volume 241 of {\em
  Graduate Texts in Mathematics}.
\newblock Springer, New York, 2007.

\bibitem{sokalsurvey}
Alan~D. Sokal.
\newblock The multivariate {T}utte polynomial (alias {P}otts model) for graphs
  and matroids.
\newblock In {\em Surveys in combinatorics 2005}, volume 327 of {\em London
  Math. Soc. Lecture Note Ser.}, pages 173--226. Cambridge Univ. Press,
  Cambridge, 2005.

\bibitem{MR2608338}
XiaoGuang Wang, WeiYuan Qiu, YongCheng Yin, JianYong Qiao, and JunYang Gao.
\newblock Connectivity of the {M}andelbrot set for the family of
  renormalization transformations.
\newblock {\em Sci. China Math.}, 53(3):849--862, 2010.

\bibitem{WELSH}
D.~J.~A. Welsh.
\newblock {\em Matroid theory}.
\newblock Academic Press [Harcourt Brace Jovanovich, Publishers], London-New
  York, 1976.
\newblock L. M. S. Monographs, No. 8.

\bibitem{WUSURVEY}
F.~Y. Wu.
\newblock The potts model.
\newblock {\em Rev. Mod. Phys.}, 54:235--268, Jan 1982.

\bibitem{MR3153590}
Fei Yang and Jinsong Zeng.
\newblock On the dynamics of a family of generated renormalization
  transformations.
\newblock {\em J. Math. Anal. Appl.}, 413(1):361--377, 2014.

\end{thebibliography}
\vspace{0.2in}

\end{document}